\documentclass[12pt, reqno]{amsart}
\usepackage{fullpage}
\usepackage{amscd, amsaddr, amsmath}
\usepackage{tikz, tikz-cd, extpfeil}
\usepackage{amsmath, enumitem, hyperref}
\usepackage{fullpage}

\usepackage{color}
\usepackage{hyperref}
\usepackage{wasysym}
\usepackage[all]{xy}
\usepackage{textgreek}

\newtheorem{Theorem}{Theorem}[section]
\newtheorem{Lemma}[Theorem]{Lemma}
\newtheorem{Proposition}[Theorem]{Proposition}
\newtheorem{Corollary}[Theorem]{Corollary}

\theoremstyle{definition}
\newtheorem{Definition}[Theorem]{Definition}

\theoremstyle{remark}
\newtheorem{Remark}[Theorem]{Remark}

\newcommand{\ZZ}{\mathbb{Z}}

\newcommand{\be}{\boldsymbol{e}}

\newcommand{\g}{\mathfrak{g}}

\newcommand{\stab}{\mathfrak{stab}}
\newcommand{\CSpin}{\mathrm{CSpin}}

\newcommand{\Hom}{\operatorname{Hom}}

\renewcommand{\mod}{\operatorname{mod}}

\newcommand{\SO}{\operatorname{SO}}

\newcommand{\Id}{\operatorname{Id}}

\newcommand{\fso}{\mathfrak{so}}
\newcommand{\Cl}{C\ell}
\newcommand{\RR}{\mathbb{R}}
\newcommand{\vol}{\operatorname{vol}}
\newcommand{\fp}{\mathfrak{p}}

\newcommand{\fg}{\mathfrak{g}}
\newcommand{\fh}{\mathfrak{h}}

\newcommand{\fk}{\mathfrak{k}}

\newcommand{\fa}{\mathfrak{a}}
\newcommand{\fD}{\mathfrak{D}}
\newcommand{\Spin}{\mathrm{Spin}}

\DeclareMathOperator{\Ker}{Ker}
\newcommand{\symb}{\operatorname{symb}}

\title{Some rigidity results for supergravity\\ backgrounds in $11$ dimensions}
\author[E. Di Bella]{Emanuele Di Bella}
\address{Universit\`a di Trento, Dipartimento di Matematica, \\
 Via Sommarive 14, 38123 Povo TN (Italy)}
\email{emanuele.dibella@unitn.it}
\author[W. A. de Graaf]{Willem Adriaan De Graaf}
\address{Universit\`a di Trento, Dipartimento di Matematica, \\
 Via Sommarive 14, 38123 Povo TN (Italy)}
\email{willem.degraaf@unitn.it}
\author[A. Santi]{Andrea Santi}
\address{Universit\`a di Roma Tor Vergata, Dipartimento di Matematica,\\ Via della Ricerca Scientifica 1, 00133  Roma (Italy)}
\email{santi@mat.uniroma2.it}
\date{\today}

\begin{document} 

\begin{abstract}
This paper is a contribution to the supersymmetry gap problem for supergravity backgrounds $(M,g,F)$ in $11$ dimensions. We study restrictions on the curvature of 
$(M,g,F)$ and, using the bijective correspondence between the space of certain filtered deformations of Lie superalgebras and the space of
highly supersymmetric supergravity backgrounds,
we establish  the following general rigidity result: if the $4$-form $F$ has rank $\operatorname{rk}(F)\leq 6$, Euclidean support, and the space  $\mathfrak{k}_{\bar 1}$ of Killing spinors has dimension $\dim\mathfrak{k}_{\bar 1}> 26$ then $(M,g,F)$ is locally isometric to the maximally supersymmetric Minkowski spacetime or Freund–Rubin background $\mathrm{AdS}_7\times\mathrm{S}^4$. The same rigidity result but with finer estimates on $\dim\mathfrak{k}_{\bar 1}$ is provided for certain types of $\mathfrak k_{\bar 1}$ and specific orbits of the $4$-form under the action of the Lorentz group.
\end{abstract}
\maketitle
\tableofcontents
\section{Introduction}
\label{sec:introduction}

Supergravity is the supersymmetric extension of Einstein's theory of General Relativity, having many distinct presentations according to, for instance, the dimension of the associated Lorentzian manifolds and the amount of supersymmetry generators. Among all supergravity theories, the most remarkable one
is arguably supergravity in $D=11$ dimensions, which was first constructed by Cremmer, Julia,
and Scherk in 1978 \cite{CJS}. One of the major contributors to the interest in $11$-dimensional supergravity today is the fact that it arises as ``low energy limit'' of M-theory -- a conjectured
unification of all the string theories \cite{W}.  Due to this, it has always been desirable to
understand properties of the supergravity theories, clarifying and generalizing their constructions from various points of view and using different techniques. Since the relevant literature is rather extensive (especially in $D=11$ dimensions), we have to refrain from a comprehensive list, and refer here to the recent \cite{Handbook}, which includes sections dedicated to
superspace \cite{KRT}, exceptional field theory \cite{Samtleben}, and group manifold \cite{Castellani} approaches. It is also relevant to mention the interpretation of supergravity as generalized geometry \cite{CSW}, via pure spinors \cite{Cederwall}, and the gauge enhancement of branes as considered in \cite{Bra}.  This paper 
is a contribution to $D=11$ supergravity from the point of view of Tanaka structures \cite{Tan} and filtered deformations of Lie superalgebras \cite{CK}, as introduced by J. Figueroa-O'Farrill and the third named author first in \cite{FOFS2017, FOFS2017II} (see also the earlier \cite{FOFMP, SSI}). For further recent results in $D=11$ supergravity that are connected to the theory of Tanaka structures, see \cite{H, H2}.

An important line of research in string and M-theory is the construction of backgrounds of their low-energy effective counterparts, i.e., supergravity, since this is where 
strings and higher-dimensional branes can propagate \cite{Blau,FOFP,FOFP2}.
The backgrounds preserving maximal or near to maximal supersymmetry are remarkable classes of such solutions, especially for carrying out the quantization of strings or branes. From a geometric perspective, a (bosonic) background of $D=11$ supergravity is 
an $11$-dimensional Lorentzian spin manifold $(M,g)$ with a closed $4$-form $F\in\Omega^4(M)$, satisfying the following
coupled system of PDEs:
\begin{equation}
    \label{eq:bosfieldeqs}
    \begin{split}
      d\star F &= \tfrac12 F\wedge F\;,\\
      \operatorname{Ric}(X,Y) &=\tfrac12 g(\imath_XF,\imath_Y F) - \tfrac16
      \|F\|^2 g(X,Y)\;,
    \end{split}
  \end{equation}
  for all $X,Y\in\mathfrak{X}(M)$. 
	If $S(M)\to M$ is the spinor bundle over $M$, with typical fiber $S\cong\mathbb R^{32}$, a spinor field $\varepsilon\in\Gamma(S(M))$ is called a {\it Killing spinor} if 
	\begin{equation}
	\label{eq:Killingspinor-equation}
	\nabla_X \varepsilon-\tfrac1{24} (X \cdot F - 3 F \cdot X)\cdot\varepsilon=0
	\end{equation}
for all $X\in\mathfrak X(M)$, where $\nabla$ is the Levi-Civita connection and $\cdot$ Clifford multiplication. The amount $N$ of linearly independent Killing spinors is an important invariant of backgrounds.

Backgrounds with $16<N\leq  32$ Killing spinors are called {\it highly supersymmetric} and their classification 
has been pursued on and off during the last 30 years. Although their full classification is still a distant goal, a number of different
solutions have been found:
\vskip0.1cm\par\noindent
\begin{itemize}
	\item[(i)] {\it Maximally supersymmetric backgrounds} (i.e., $N = 32$) were classified in \cite{FOFP} up to local isometry (see \cite{FOFS2017} for a purely Lie theoretic proof in the spirit of this paper):
	\begin{align*}
	\;\;\;\;\;\;\;
	{\footnotesize
 \begin{array}{|c|c|c|c|} \hline
\text{\it Name}  &\text{\it Background} & \varphi:=F^\sharp|_o & \text{\it Killing superalgebra}  \\ \hline\hline
\text{Freund--Rubin} & \begin{gathered} \\ Ad_4\times S^{7} \\ \\ \end{gathered} & \be_{0123} & \mathfrak{osp}(8|4) \\ \hline
\text{Freund--Rubin} & \begin{gathered} \\ Ad_7\times S^4 \\ \\ \end{gathered} & \be_{1234} & \mathfrak{osp}(2,6|4) \\ \hline
\text{Kowalski-Glikman} & \begin{gathered} \\ \text{Cahen-Wallach symmetric space} \\ \\ \end{gathered} & \be_{+123} & \text{solvable Lie superalgebra} \\ \hline
\text{Flat background} &	\begin{gathered} \\ \text{Minkowski spacetime}\; \\ \\ \end{gathered} V= \mathbb R^{1,10} & 0 & \text{Poincar\'e superalgebra}\;\fp \\ \hline
\end{array}
}
\end{align*}
	\item[(ii)] {\it pp-waves backgrounds}, i.e., Brinkmann spaces with flat transverse geometry,
	\item[(iii)] {\it G\"odel backgrounds}, i.e., solutions admitting closed timelike curves.
\end{itemize}

A supersymmetry gap result is also known:
\begin{Theorem}
\cite{GGPR, GGP}
\label{thm:No-Go31}
If a background of $D=11$ supergravity has at least $30$ Killing spinors, then it is locally isometric to a maximally supersymmetric background.
\end{Theorem}
The proof relies on a careful analysis of integrability conditions on the curvature of the connection \eqref{eq:Killingspinor-equation}, field equations     \eqref{eq:bosfieldeqs}, Bianchi identities of the Riemann curvature, and $dF=0$. A somewhat shorter and purely Lie theoretic proof for $N=31$ can also be found in \cite{S}. Unfortunately, extending the strategy of \cite{GGPR, GGP} to $N=29$ and below is hindered by the fact 
that the orbit structure of the action of $\Spin(V)$ on the Grassmannian $\mathrm{Gr}(k,S)$ of $k$-planes in the spinor module $S$ is already extremely complicated for $k=3$. 

The primary motivation for this article is to shed some light on the so-called {\it supersymmetry gap problem}, namely the determination of the submaximal supersymmetry dimension. Although different in several respects, this ``supersymmetry gap'' is reminiscent of the gap phenomenon for
classical geometric structures (we refer to \cite{KT} and the references therein).
At present, the highest number of Killing spinors
known for non-maximally supersymmetric backgrounds is $N=26$, reached by 
the pp-wave discovered by J. Michelson in \cite{M}.
Highly supersymmetric pp-waves exist with any even number $18\leq N\leq 24$ of Killing spinors \cite{GH}, whereas $N=20$  or $N=18$ 
for G\"odel backgrounds \cite{HT}. We argue that the rank of the $4$-form can be considered as a useful organizing principle to study the supersymmetry gap problem (instead of the Grassmannian $\mathrm{Gr}(k,S)$ for varying values of $k$): to date, there is no known highly supersymmetric background for which the rank is maximal and known backgrounds with $N\geq 24$ have rank $\operatorname{rk}(F)\leq 8$. 

\vskip0.2cm\par
The goal of this paper is to establish the following rigidity result.
\begin{Theorem}
\label{thm:main}
Let $(M,g,F)$ be a background of $D=11$ supergravity with 
$4$-form $F$ having rank $\operatorname{rk}(F)\leq 6$ and Euclidean support. If the space of Killing spinors has dimension $N> 26$, then 
$(M,g,F)$ is locally isometric to the maximally supersymmetric flat background or to the Freund–Rubin background $AdS_7\times S^4$.
\end{Theorem}
Some other rigidity results have recently appeared in the $D=11$ supergravity literature, but under working assumptions that the spacetime $(M,g)$ and the $4$-form $F$ are factorized, and for supergravity backgrounds in low supersymmetric regime (see for example \cite{Farotti, Fei}). 

The proof of Theorem \ref{thm:main} builds upon
general structural results for highly-supersymmetric backgrounds. We summarize their main aspects here and refer to \S\ref{sec:2} and \S\ref{sec:3.1} for more details.
\begin{Theorem} 
\label{thm:local-homogeneity}
\label{thm:JA}
\cite{FOFH, FOFMP, FOFS2017, FOFS2017II, Abel} Let $(M,g)$ be an $11$-dimensional Lorentzian spin manifold 
endowed  with a closed $4$-form $F\in\Omega^4(M)$. Then:
\begin{enumerate}
\item There exists an associated Lie superalgebra $\fk=\fk_{\bar 0}\oplus \fk_{\bar 1}$, called Killing superalgebra,
with $\fk_{\bar 0}$ the space of Killing vectors preserving $F$ and $\fk_{\bar 1}$ the space of Killing spinors;
\end{enumerate}
Moreover, if $(M,g,F)$ is highly supersymmetric, then:
	\begin{enumerate}		
		\item[(2)] $(M,g,F)$ is locally
homogeneous;
	\item[(3)] The bosonic field
  equations    \eqref{eq:bosfieldeqs} of $D=11$ supergravity are automatically satisfied;	
	\item[(4)] The Killing superalgebra $\fk$ is isomorphic to a filtered subdeformation $\mathfrak{g}$ of the Poincar\'e superalgebra $\mathfrak p$ and the same result holds for the so-called transvection superalgebra (the ideal of $\mathfrak k $ generated by its odd part $\fk_{\bar 1}$);
\item[(5)] The background is fully determined, up to local isometry, by its associated geometric symbol, namely by the pair $\symb(M,g,F)=(\varphi,S')$, where 
\begin{equation}
\label{eq:pair-introduction}
\varphi=F^\sharp|_o\in\Lambda^4 V\quad\text{and}\quad S'=\fk_{\bar 1}|_o\subset S\;.
\end{equation}
In particular if $\varphi=0$ then $(M,g)$ is locally isometric to Minkowski spacetime.
		\end{enumerate}
\end{Theorem}

The proof of Theorem \ref{thm:main} takes 
$\SO(V)$-orbits of four-vectors $\varphi\in\Lambda^4 V$ as starting point. 
We will resort to one of the fundamental properties satisfied by the stabilizer algebra of the transvection superalgebra of the background and a knowledge of the space  of bispinors $\odot^2 S'$ that is mostly independent of $S'$ (more precisely, it is sufficient to know only its dimension). The crucial conditions for the proof will be provided by the Lie brackets of the transvection superalgebra corresponding to the curvature in superspace on purely fermionic directions. 

To conclude, we would like to emphasize that most of our arguments will be representation-theoretic (in particular, we avoid the use of Fierz Identities), thus more amenable to generalizations. In fact, we completed the classification of semisimple and nilpotent {\it real}
fourvectors in dimension $8$ under the action of the special linear group  in the recent \cite{BGS} -- this result comprises the classification of the fourvectors of rank $\leq 7$ (in any dimension), which will be particularly useful for future applications along the lines of this article.
\par
\bigskip
The paper is organized as follows. In \S\ref{sec:2} we review  the basics of filtered subdeformations of the Poincar\'e superalgebra, including the strong version Theorem \ref{thm:strong-version-reconstruction} of the Reconstruction Theorem 
and the concepts of Dirac kernel and Lie pair.
The latter is a pair $(\varphi,S')$ satisfying a system of coupled algebraic equations, quadratic on $\varphi$ and cubic on $S'$, which allows to recover the stabilizer algebra of the transvection Lie superalgebra.
In \S\ref{sec:3}, we first overview the general strategy of the proof of Theorem \ref{thm:main}, turn to study $\mathrm{SO}(V)$-orbits in $\Lambda^4 V$ with small rank and Euclidean support and, finally, reduce the problem to rank $6$ orbits via Theorem \ref{prop:sugra-decomposable}. The interplay of spinors and fourvectors via the superspace curvature is carefully analyzed in \S\ref{sec:4} and Theorem \ref{thm:main} is then established in \S\ref{sec:5}-\S\ref{sec:6}: it is the combination of Theorem \ref{mainthm:1} in \S\ref{sec:5} (for orbits of lenght $3$) and  Theorem \ref{mainthm:2} in \S\ref{sec:6} (for orbits of lenght $2$).

\par\bigskip
\noindent{\it Notations.} Throughout  the paper,  we consider  Clifford algebras as defined  in \cite{LM}.  According to this,  the Clifford product  of vectors   of an orthonormal basis of a pseudo-Euclidean vector space $(V,\eta)$ is given by $\be_i \cdot \be_j +\be_j\cdot\be_i= -  2 \eta_{ij}$ and  not  `` $+  2 \eta_{ij}$'' as it is  sometimes assumed. We also follow the conventions of \cite{LM} to tacitly identify $\mathfrak{so}(V)$ with $\Lambda^2V$, namely, we have $v\wedge w(u):=\eta(v,u)w-\eta(w,u)v$ for all $u,v,w\in V$, and the action of $\mathfrak{so}(V)$ on any spinor module is half the Clifford action of $\Lambda^2V$.

\section{Preliminary notions}
\label{sec:2}
\subsection{The Poincar\'e superalgebra, its filtered subdeformations, and realizability}
\label{sec:2.1}
Let $(V,\eta)$ be the Lorentzian vector space with ``mostly minus'' signature $(1,10)$
and $S$ the spinor representation of $\fso(V)$ 
(irreducible module of the Clifford
algebra $\Cl(V)\cong 2\RR(32)$ for which, in our conventions, the volume $\vol_V\in\Cl(V)$ acts as
$\vol_V\cdot s=-s$ for all $s\in S$. For more details on Clifford algebras and their representations, we refer the reader to \cite{LM, AC1, AC2}). For later use in \S\ref{sec:3}-\S\ref{sec:4}, we fix an orthonormal basis $\left\{\be_0,\be_1,\ldots,\be_9,\be_\natural\right\}$ of $V$ once and for all, where all the vectors of the basis have squared norm $-1$, except $\be_0$ for which $\eta(\be_0,\be_0)=+1$.
We recall that $S$ has an
$\fso(V)$-invariant symplectic structure $\left<-,-\right>$ with the property that
$
\left<v\cdot s_1, s_2\right> = - \left<s_1, v \cdot s_2\right>
$
w.r.t. the Clifford action, for all $s_1,s_2 \in S$ and $v \in V$. 
\begin{Definition}
The {\it Poincar\'e superalgebra} $\fp$ has underlying vector space
$\fso(V) \oplus S \oplus V$ and nonzero Lie brackets given by
\begin{equation}
\label{eq:PSA}
    [A,B] = AB - BA\;, \qquad [A,s] = As\;, \qquad [A,v] = Av\;, \qquad
    [s,s] = \kappa(s,s)\;,
\end{equation}
where $A,B \in \fso(V)$, $v \in V$ and $s \in S$. 
\end{Definition} Here
$\kappa: \odot^2 S \to V$ is the so-called \emph{Dirac current},
defined by
$\eta(v,\kappa(s,s)) = \left<s, v\cdot s\right>$,
for all $s\in S$ and $v\in V$.  One important property is that the restriction of $\kappa$ to
$\odot^2 S'$ is surjective on $V$, for any subspace $S'\subset S$ with $\dim S' > 16$ \cite{FOFH}. 
This algebraic fact is usually referred to as the ``local Homogeneity Theorem'', due to the
role it plays proving (2) of Theorem \ref{thm:local-homogeneity}.

  If we grade $\fp$ so that $\fso(V)$, $S$ and $V$ have degrees
$0$, $-1$ and $-2$, respectively, then the Lie brackets \eqref{eq:PSA} turn
$\fp$ into a $\ZZ$-graded Lie superalgebra
$\fp = \fp_{0} \oplus \fp_{-1} \oplus \fp_{-2}=\fso(V)\oplus
S\oplus V$.
Let $\fa=\fh\oplus S'\oplus V$ be a graded subalgebra of $\fp$ with $\dim S'>16$ and $\fg$ a {\it filtered deformation} of $\fa$, i.e., the Lie brackets of $\fg$ take the following general form \cite{CK}:
\begin{equation}
\label{eq:generalbrackets}
  \begin{aligned}
    [A,B]&=AB-BA\\
    [A,s]&=As\\
    [A,v]&=Av+\delta(A,v)
  \end{aligned}
  \qquad\qquad
  \begin{aligned}
    [s,s]&=\kappa(s,s)+\gamma(s,s)\\
    [v,s]&=\beta(v,s)\\
    [v,w]&=\alpha(v,w)+\rho(v,w)
  \end{aligned}
\end{equation}
for $A,B\in \fh$, $v,w\in V$, $s\in S'$. Here the maps $\alpha\in\Hom(\Lambda^2 V,V)$,
$\beta\in\Hom(V\otimes S',S')$, $\gamma\in\Hom(\odot^2 S',\fh)$ and
$\delta\in\Hom(\fh\otimes V,\fh)$ have degree $2$,
and $\rho\in\Hom(\wedge^2V,\fh)$ has degree $4$. 

If we do not  mention the subalgebra $\fa$ explicitly, we say that
$\fg$ is a \emph{highly supersymmetric filtered subdeformation} of $\fp$. Finally, we say that 
$\fg=\fg_{\bar 0}\oplus\fg_{\bar 1}$ is {\it odd-generated} if $\fg_{\bar 0}=[\fg_{\bar 1},\fg_{\bar 1}]$. 
The notion of isomorphism $\Phi:\fg\to\widetilde\fg$ of highly supersymmetric filtered subdeformations of $\fp$ can be found in \cite[Def. 5]{FOFS2017II}: 
it suffices to know that
\begin{equation}
  \label{eq:generaliso}
  \Phi(A)=g\cdot A\;,\qquad
  \Phi(s)=g\cdot s\;,\qquad\text{and}\qquad
  \Phi(v)=g\cdot v+ X_v\;,
\end{equation}
for some $g\in\Spin(V)$ and $X:V\to\widetilde\fh$. The notion of embedding is in \cite[Def. 11]{FOFS2017II}.

We are interested in filtered subdeformations that are {\it realizable}, that is, which correspnd to some highly supersymmetric $D=11$ supergravity background. To introduce this concept, we first need two important maps.
\begin{Definition}
Associated to any $\varphi\in\Lambda^4 V$, there are two natural maps $\beta^\varphi:V\otimes S\to S$ and
$\gamma^\varphi:\odot^2 S\to\fso(V)$ defined by
\begin{align}
\label{eq:thetatoo}
\beta^{\varphi}(v,s)&=\tfrac1{24}(v\cdot\varphi-3\varphi\cdot v)\cdot s\;,\\ 
\label{eq:thetatoo-II}
\gamma^{\varphi}(s,s)v&=-2\kappa(\beta^\varphi(v,s),s)\;,
\end{align}
for all $s\in S$, $v\in V$. We will sometimes also use the notation $\beta^\varphi_v(s) = \beta^\varphi(v,s)$. 
\end{Definition}
As explained in \cite{FOFS2017}, these maps 
are the components of normalized cocycles for the Spencer cohomology group $H^{2,2}(\fp_-,\fp)\cong\Lambda^4 V$. From a more geometric perspective, $\beta^\varphi$ encodes  the Killing spinor equations for supergravity in 
$D=11$ (this fact has been exploited also for other supersymmetric field theories in different dimensions, see \cite{MFSI, MFSII, FoFF, B1, B0} for details) while $\gamma^\varphi$ is related to the curvature in superspace acting on purely fermionic directions.
\begin{Definition}\cite{FOFS2017II}
\label{def:realizable}
A highly supersymmetric filtered subdeformation $\fg$ of $\fp$ is called {\it realizable} if there exists $\varphi\in\Lambda^4 V$ such that:
\begin{enumerate}
	\item[(i)] $\varphi$ is $\fh$-invariant;
	\item[(ii)] $\varphi$ is closed, i.e., 
	\begin{equation}
	\label{eq:closed}
	d\varphi(v_0,\ldots,v_4) = \sum_{i<j}(-1)^{i+j}
      \varphi(\alpha(v_i,v_j), v_0, \ldots, \hat v_i,\ldots,\hat
      v_j,\ldots,v_4) = 0
    \end{equation}
    for all $v_0,\ldots, v_4\in V$;
	\item[(iii)] The components of the Lie brackets of $\fg$ of degree $2$ are of the form
	\begin{equation}
  \label{eq:KSAina}
  \begin{aligned}
	 \alpha(v,w) &= X_v w - X_w v\\
   \gamma(s,s) &= \gamma^\varphi(s,s) - X_{\kappa(s,s)}
  \end{aligned}
  \qquad\qquad
  \begin{aligned}  
	\beta(v,s) &= \beta^\varphi(v,s) + X_v s\\
  \delta(A,v) &= [A,X_v] - X_{Av}  
\end{aligned}
\end{equation}
for some linear map $X:V\to\fso(V)$, where $A,B \in \fh$, $v,w \in V$ and $s \in
S'$.
\end{enumerate}
\end{Definition}
It is a non-trivial result that a realizable $\fg$ has a {\it unique} $\varphi\in\Lambda^4 V$ 
that satisfies (i)-(iii) of Definition \ref{def:realizable} \cite[Corollary 26]{FOFS2017II} and that the component $\rho$ of degree $4$ is fully determined by those of degree $2$ \cite[Proposition 6]{FOFS2017II}.
(We will explain it in \S\ref{sec:2.2}. The idea is that $\varphi$ is an element of $H^{2,2}(\fa_-,\fa)$ and, since $H^{4,2}(\fa_-,\fa)=0$, this group determines the deformation.) 
Changing the map $X:V\to\fso(V)$ by some values in $\fh$ gives filtered subdeformations that are isomorphic.
\vskip0.3cm\par\noindent

It is well-known that 
$\odot^2 S\cong\Lambda^1 V\oplus \Lambda^2 V\oplus\Lambda^5 V$ 
as an $\fso(V)$-module. This decomposition is unique, since all 
three summands are $\fso(V)$-irreducible and inequivalent, so we will
consider $\Lambda^q V$ directly as a
subspace of $\odot^2 S$, for $q=1,2,5$.
We decompose any $\omega\in\odot^2 S$ accordingly
\vskip0.3cm\par\noindent
$$\omega=-\tfrac{1}{32}\big(\omega^{(1)}+\omega^{(2)}+\omega^{(5)}\big)\;,$$ 
\vskip0.2cm\par\noindent
where the overall factor 
is introduced so that $\omega^{(1)}$
coincides exactly
with the Dirac current of $\omega$ and the $\omega^{(q)}\in\Lambda^q V$ for $q=2,5$ are defined in a completely similar fashion via polyvectors. 

We may then re-write $\gamma^\varphi:\odot^2 S\to\mathfrak{so}(V)$ in a couple of useful auxiliary ways (here musical isomorphisms have been tacitly used; this will also
be the case throughout the whole paper):
\begin{Lemma}
\label{lem:mitoccaaggiungere}
For all $s,t\in S$, $\omega\in\odot^2S$, $v,w\in V$, we have
\begin{equation}
\label{definition-gamma-useful}
\begin{aligned}
\eta(\gamma^\varphi(\omega)v,w)
&=\tfrac{1}{3}\eta(\imath_v\imath_w\varphi,\omega^{(2)})
+\tfrac{1}{6}\eta(\imath_v\imath_w\star\varphi,\omega^{(5)})\;,
\end{aligned}
\end{equation}
\begin{equation}
\label{lem:reformulation}
\eta(\gamma^\varphi(s,t)v,w)=-\tfrac{1}{12}\langle\varphi\cdot s,v\wedge w \cdot t\rangle
-\tfrac{1}{12}\langle\varphi\cdot t,v\wedge w \cdot s\rangle
+\tfrac12\langle s,\imath_v\imath_w\varphi\cdot t\rangle\;,
\end{equation}
where $\star$ is the Hodge star operator. In particular $\Ker\gamma^\varphi\supset\Lambda^1 V$.
\end{Lemma}
\begin{proof}
The first identity follows from taking the scalar product of \eqref{eq:thetatoo-II} with $w$ and rewriting 
$\beta^{\varphi}(v,s)=-\tfrac{1}{12}(v\wedge\varphi+2\imath_v\varphi)\cdot s$ 
so to get
\begin{equation*}
\begin{aligned}
\eta(\gamma^\varphi(\omega)v,w)&=-2\langle s,w\cdot \beta^{\varphi}(v,s)\rangle\\
&=\tfrac{1}{6}\langle s, \big(2\imath_v\imath_w\varphi-v\wedge w\wedge \varphi\big)\cdot s\rangle   \\
&=\tfrac{1}{3}\eta(\imath_v\imath_w\varphi,\omega^{(2)})
+\tfrac{1}{6}\eta(\imath_v\imath_w\star\varphi,\omega^{(5)})\;.
\end{aligned}
\end{equation*}
For the second identity, since both expressions are symmetric in $s,t$, it is sufficient to assume $s=t$ and remember that 
$-\tfrac{1}{6}\langle\varphi\cdot s,v\wedge w \cdot s\rangle=\tfrac{1}{6}\langle v\wedge w \cdot\varphi\cdot s, s\rangle=\tfrac{1}{6}\langle v\wedge w \wedge\varphi\cdot s, s\rangle+\tfrac{1}{6}\langle \imath_v\imath_w \varphi\cdot s, s\rangle$. Summing up with the last term on the r.h.s. of \eqref{lem:reformulation} with $s=t$, we arrive at equation \eqref{definition-gamma-useful}.
\end{proof}

\subsection{The Reconstruction Theorem and the Dirac kernel}
\label{sec:2.2}
\begin{Theorem}[Reconstruction Theorem]\cite{FOFS2017II, Abel}
\label{thm:reconstruction}
\begin{enumerate}
	\item The assignment that sends a highly supersymmetric supergravity background to its transvection superalgebra is a $1:1$ correspondence between the moduli spaces
	\vskip0.1cm\par\noindent
\begin{align*}
	\;\;\;\;\;\;\;\frac{\mathcal{SB}}{\cong}&=\frac{\big\{\text{highly supersymmetric supergravity bkgds}\; (M,g,F)\big\}}{\text{local isometry}}\\
	\frac{\mathcal{FD}}{\cong}&=\frac{\big\{\text{highly supersymmetric, odd-generated, realizable, filtered subdeformations $\mathfrak g$ of $\mathfrak p$}\big\}}{\text{maximality and isomorphism}}
	\end{align*} 
	\vskip0.1cm\par\noindent
	of the highly supersymmetric supergravity backgrounds and the
highly supersymmetric, odd-generated, realizable, filtered subdeformations (which are maximal in this class).
\item The curvature $R:\Lambda^2 V\to\fso(V)$ of the supergravity background 	$(M,g,F)\in\mathcal{SB}$ associated to $\fg\in\mathcal{FD}$ is given by 
$R(v,w)=\rho(v,w)-[X_v,X_w]+X_{\alpha(v,w)}$, for $v,w\in V$.
\end{enumerate}
\end{Theorem}
Let $\fg\in\mathcal{FD}$ with the corresponding graded Lie algebra $\fa=\fh\oplus S'\oplus V$ and $\varphi\in\Lambda^4 V$. Due to Theorem \ref{thm:reconstruction}, we may call $(\varphi,S')$ the symbol of $\fg$.
Since $S'$ has dimension $\dim S'>16$, then
$\odot^2 S'\subset \odot^2 S\cong\Lambda^1V\oplus\Lambda^2
  V\oplus\Lambda^5 V$ projects surjectively to $\Lambda^1 V$
via the Dirac current.  The above embedding
is in general diagonal and one cannot expect  $\odot^2 S'$ to contain
$\Lambda^q V$, not even for $q=1$.
Restricting the
Dirac current to $\odot^2 S'$ gives rise to
a short exact sequence
\begin{equation*}
  \begin{CD} 0 @>>> \fD @>>> \odot^2 S' @>\kappa>> V @>>> 0\;.
  \end{CD}
\end{equation*}
\begin{Definition}
The space
\begin{equation}
\label{eq:Dirac-kernel-definition}
  \begin{split} \mathfrak{D} &= \odot^2 S'\cap (\Lambda^2
    V\oplus\Lambda^5 V)=\left\{\omega\in\odot^2
      S'~\middle|~\omega^{(1)}=0\right\}
  \end{split}
\end{equation}
is called the \emph{Dirac kernel} of $S'$. 
\end{Definition}
The Dirac kernel will play a major role in the proof of our main Theorem \ref{thm:main}. To explain why this is the case, we need a few last concepts and preliminary results. A splitting of the above short exact sequence --- that is, a linear map $\Sigma:V\to \odot^2 S'$ such that $\Sigma(v)^{(1)}=v$ for all
$v\in V$ --- is called a \emph{section} associated to $S'$. A section
associated to $S'$ always exists and it is unique up to
elements in the Dirac kernel. Finally, we set
  \begin{equation}
	\label{isotropy-phi}
    \begin{split}
      \fh_{(\varphi,S')} &=\gamma^{\varphi}(\mathfrak{D})= \left\{\gamma^{\varphi}(\omega)~\middle|~\omega\in\odot^2
        S'~\text{with}~\omega^{(1)}=0\right\}\;,
    \end{split}
  \end{equation}
	which is a subspace of $\fso(V)$.
\begin{Definition}
  \label{def:Lie pair}	
	\label{eq:Lie-pair}
The symbol $(\varphi,S')$ is called a \emph{Lie pair} if $\fh_{(\varphi,S')}\subset\mathfrak{stab}_{\fso(V)}(\varphi)\cap \mathfrak{stab}_{\fso(V)}(S')$.
\end{Definition}
The name ``Lie pair'' is motivated by the fact that the corresponding \eqref{isotropy-phi} is, in that case, a
  Lie subalgebra of $\fso(V)$ \cite[Lemma 18]{FOFS2017II}. Note that the defining equations of a Lie pair is a rather complicated system of coupled algebraic equations, quadratic on $\varphi$ and cubic on $S'$.
	The following results give crucial necessary conditions satisfied by
any $\fg\in\mathcal{FD}$. 
\begin{Proposition}\cite{FOFS2017II}
  \label{prop:LiePair}
Let $\fg\in\mathcal{FD}$
  with the underlying graded Lie algebra $\fa=\fh\oplus S'\oplus V$. Then the symbol $(\varphi,S')$ is a Lie pair and:
  \begin{enumerate}
  \item the stabilizer Lie algebra $\fh=\fh_{(\varphi,S')}$;
  \item the map $X:V\to\fso(V)$ is determined, up to elements in
    $\fh$, by the identity
    \begin{equation}
		\label{eq:map-X}
      X=\gamma^\varphi\circ\Sigma\;,
    \end{equation}
    where $\Sigma$ is any section associated to $S'$.
  \end{enumerate}  
  In particular $\fg$ is fully determined, up to isomorphism, by the associated symbol. 
\end{Proposition}
\begin{Corollary}
\label{corollary:geometric-symbol}
Any highly supersymmetric supergravity background $(M,g,F)$ is completely determined by its geometric symbol $\symb(M,g,F)$, up to local isometry.
\end{Corollary}
It is not true that every Lie pair $(\varphi,S')$ has a corresponding $\g\in\mathcal{FD}$ or, in other words, it is the geometric symbol of a supergravity background. In fact,
the Lie brackets of a highly supersymmetric, odd-generated, realizable, filtered subdeformation $\mathfrak g$ of $\mathfrak p$ are given by
\begin{equation}  
\label{eq:preKSA}
  \begin{aligned}
    [A,B]&=AB-BA\\
    [A,s]&=As\\
    [A,v]&=Av+[A, X_v] - X_{Av}
  \end{aligned}
  \;\;
  \begin{aligned}
    [s,s]&=\kappa(s,s) + \gamma^{\varphi}(s,s) - X_{\kappa(s,s)}\\
    [v,s]&=\beta^{\varphi}(v,s) + X_v s\\
    [v,w]&=X_v w - X_w v + [X_v,X_w] - X_{X_v w - X_w v} + R(v,w)
  \end{aligned}
\end{equation}
for all $A,B\in\fh$, $v,w\in V$, $s\in S'$. 
Here $\fh=\fh_{(\varphi,S')}$ and the map $X:V\to\fso(V)$ is as in \eqref{eq:map-X}.
\mbox{The rest of the data given by the curvature $R : \Lambda^2 V \to \fso(V)$ also depends on the Lie pair,}\\
as we now  detail explicitly.

The algebraic structure \eqref{eq:preKSA} entails further obvious constraints. First, the right-hand sides of the Lie brackets take values in the graded subalgebra $\fa=\fh \oplus S' \oplus V$ of $\fp$ underlying $\fg$
(the individual terms may not, see equations (23)-(26) in \cite{Abel}). Moreover the Lie brackets are subject to Jacobi identities -- there are ten components --  and we have to remember $d\varphi=0$. 

The overall system of equations is involved and it has been carefully analyzed in \S 5 of \cite{Abel}. 
Theorem 5 op. cit. drastically simplifies the situation and leads to following Theorem
\ref{thm:strong-version-reconstruction} -- a stronger version  of the Reconstruction Theorem. Therein 
\vskip0.4cm\par\noindent
\begin{equation*}
\label{eq:moduli-abstract-symbols}
\begin{aligned}
\frac{\mathcal{AS}}{\cong}&=\frac{\left\{\text{abstract symbols}\;(\varphi,S')\right\}}{\text{maximality and isomorphism}}
\end{aligned}
\end{equation*}
\vskip0.2cm\par\noindent
is the moduli space of abstract symbols\footnote{The isomorphism condition refers to the natural action of $\mathrm{Spin}(V)$ on pairs $(\varphi,S')$. Maximality was used explicitly in \cite{Abel} but inadvertently omitted in Theorem 5: we say $(\varphi',S')\subset (\varphi'',S'')$ when $\varphi'=\varphi''$, $S'\subset S''$.}, where a Lie pair $(\varphi,S')$ 
is called an {\it abstract symbol} if 
\begin{align}
\label{eq:symbolII}
\beta^{\varphi}(v,s) + X_v s&\in S'\\
\label{eq:symbolIII}
d\varphi(v_0,\ldots,v_4) &= 0\\
\label{eq:symbolIV}
    \tfrac12 R\big(v,\kappa(s,s)\big)w &=\kappa((X_v\beta^{\varphi})(w,s),s)
	-\kappa(\beta^\varphi_v(s) , \beta^\varphi_w(s)) -
    \kappa(\beta^\varphi_w \beta^\varphi_v(s),s)
		\\
		\label{eq:symbolV}
		R(v,w)s &= (X_v\beta^{\varphi})(w,s) - (X_w\beta^{\varphi})(v,s) +
  [\beta^\varphi_v,\beta^\varphi_w](s)
\end{align}
\vskip0.2cm\par\noindent
for some $R:\Lambda^2 V\to \fso(V)$ and all $v,w, v_0,\ldots, v_4\in V$, $s\in S'$.

\begin{Theorem}[Reconstruction Theorem - Strong Version]\cite{Abel}
\label{thm:strong-version-reconstruction}
The assignment 
$\frac{\mathcal{SB}}{\cong}\longrightarrow\frac{\mathcal{AS}}{\cong}$ that sends a highly supersymmetric supergravity background to its geometric symbol
is a $1:1$ correspondence, with image the moduli space of abstract symbols.
\end{Theorem}
It is remarkable that the curvature tensor of a $D=11$ {\it highly-supersymmetric} supergravity background $(M,g,F)$ can be expressed in terms of  the $4$-form $F$. (By the local Homogeneity Theorem, 
each of the equations \eqref{eq:symbolIV} and \eqref{eq:symbolV}
determines uniquely $R$.) {\it From this perspective, imposing restrictions on the $\mathrm{SO}(V)$-orbits of $\varphi\in\Lambda^4V$ is equivalent to studying  backgrounds that are highly-supersymmetric with special prescribed restrictions on their curvature tensors.}
\section{General strategy and first results}
\label{sec:3}
\subsection{Brief overview on general strategy}
\label{sec:3.1}
Describing $\Spin(V)$-orbits of Lie pairs $(\varphi,S')$, checking \eqref{eq:symbolII}-\eqref{eq:symbolIII}, and understanding whether \eqref{eq:symbolIV} and \eqref{eq:symbolV} anambigously define a tensor $R:
  \Lambda^2 V\to\fso(V)$ is an unfeasible task. One of the main reasons, as advertised earlier in \S\ref{sec:introduction}, is that already the orbit structure of the action of $\Spin(V)$ on the Grassmannian $\mathrm{Gr}(k,S)$ of $k$-planes in the spinor module $S$ is extremely complicated (it is known for $k=1$ \cite{B} and for $k=2$ \cite{GGP}). See \cite[\S7]{Abel} for an example explicitly carried out.
	
	However, the strategy we will pursue in \S\ref{sec:3}-\S\ref{sec:6} is different in several respects:
	\begin{enumerate}
		\item[(i)] We  take $\mathrm{SO}(V)$-orbits of fourvectors $\varphi\in\Lambda^4V$ as the starting point and their rank as a useful organizing principle;
		\item[(ii)] We set up a presentation of spinor module $S$ adapted to all relevant $\varphi$ simultaneously;
		\item[(iii)] Since the supersymmetry gap problem is our main interest, we aim to establish general rigidity results that
depend on the subspaces $S'$ of $S$ only through their dimensions. This requires a careful use of Theorem \ref{thm:strong-version-reconstruction} in \S\ref{sec:5}, since some of the conditions defining the abstract symbols are not easy to manoeuvre without an explicit knowledge of $S'$;
\item[(iv)] It is also worth pointing out that in our very recent \cite{BGS} we used Vinberg's $\theta$-groups and Galois cohomology to classify the nilpotent and semisimple {\it real} fourvectors in dimension $8$ under the action of the special linear group. This result comprises the classification of fourvectors of rank $\leq 7$, since they are automatically nilpotent, and it will be particularly useful for the future applications.
	\end{enumerate}
	\subsection{The $\mathrm{SO}(V)$-orbits in $\Lambda^4 V$ of small rank and with Euclidean support}
For the purposes of this subsection, we consider $G=\mathrm{SL}(V)$ and let $G^\theta=\mathrm{SO}(V)$ be 
the special orthogonal subgroup (the fixed point set of an appropriate involution $\theta:G\to G$ of $G$). 
The {\it support} of a fourvector $\varphi\in\Lambda^4 V$ is the unique minimal subspace $ E\subset V$ such that $\varphi\in\Lambda^4 E$. 
Its dimension is the {\it rank} of $\varphi$ and it is a $G$-invariant, in particular it is one of the simplest $G^\theta$-invariants, together with the causal type of $E$ and the fourvector's norm. The fourvectors of minimal (non-zero) rank are decomposable and constitute a $G$-orbit, which is stratified by the level sets of the norm  into a $1$-parameter family of $G^\theta$-orbits.

It is well-known that an indecomposable $\varphi\in\Lambda^4 V$ has rank at least $6$
\cite[page 103]{Martinet1970}. Since orbits of rank $\leq 7$ are automatically nilpotent from the perspective of \cite{BGS}, we directly see from  Table 1 of \cite{BGS} that the $G$-orbits of rank $6$ are those with representative $\varphi=\be_{1234}+\be_{1256}$ (length $2$ orbit) and $\varphi=\pm(\be_{1234}+\be_{1256}+\be_{3456})$ (two length $3$ orbits). These are the orbits numbered N.2, N.4, in  Table 1 of \cite{BGS}, see \cite{BGS} for an explanation of why N.4 is ``hidden'' there.
(The former  
is also a subminimal orbit, in the sense that its Zariski-closure consists of the orbit itself, the minimal orbit of non-zero decomposable fourvectors and the zero fourvector. See, for instance, \cite[page 104]{Martinet1970}.) To classify the $G^\theta$-orbits of fourvectors $\varphi\in\Lambda^4V$ of rank $6$, we determine the stratification of the above three $G$-orbits under the action of the subgroup. In this paper, we only focus on those with Euclidean support $E$.
\begin{Proposition}[$G^\theta$-orbits of fourvectors]
\label{lem:rank6fourvectors-orthogonal}
\label{lemma:invariant-4form}
\hfill
\begin{enumerate}
\item
Every $G^{\theta}$-orbit of fourvectors in $\Lambda^4V$ having rank $6$ and Euclidean support admits a unique representative
\vskip0.2cm\par\noindent
\begin{equation}
\label{eq:firstrepresentative}
\varphi_{(\rho,\lambda,\mu,\pm)}:=\pm(\rho\be_{1234}+\lambda\be_{1256}+\mu\be_{3456})\;,
\end{equation}
\vskip0.2cm\par\noindent
 where $\rho\geq \lambda\geq \mu\geq 0$ and only $\mu$ can be zero;
\item
Let $E=\langle \be_1,\ldots,\be_6\rangle$ be the support of $\varphi=\varphi_{(\rho,\lambda,\mu,\pm)}$ and 
let $E^\perp=\langle \be_0,\be_7,\be_8,\be_9,\be_\sharp\rangle$ be the orthogonal of $E$ in $V$. Then the stabiliser Lie algebra
\begin{equation}
\label{eq:stabfull}
\mathfrak{stab}_{\mathfrak{so}(V)}(\varphi)=\mathfrak{so}(E^\perp)\oplus\mathfrak{stab}_{\mathfrak{so}(E)}(\varphi)\;,
\end{equation}
where
\vskip0.2cm\par\noindent
{\footnotesize
\begin{align*}
 \;\;\;\;\;\;\;\;\;\;\begin{array}{|c|c|c|} \hline
\text{Case}  &\text{Generators of $\mathfrak{stab}_{\mathfrak{so}(E)}(\varphi)$} & \mathfrak{stab}_{\mathfrak{so}(E)}(\varphi) \\ \hline\hline
\rho=\lambda=\mu &\be_{35}+\be_{46}, \be_{36}-\be_{45},\be_{13}+\be_{24},\be_{14}-\be_{23},\be_{15}+\be_{26},\be_{16}-\be_{25}, \be_{12},\be_{34},\be_{56}& \mathfrak{u}(3) \\ \hline
  \rho>\lambda=\mu &\be_{13}+\be_{24},\be_{14}-\be_{23}, \be_{12},\be_{34},\be_{56} & \mathfrak{u}(2)\oplus\mathfrak{u}(1) \\ \hline
	 \rho=\lambda>\mu & \be_{35}+\be_{46}, \be_{36}-\be_{45}, \be_{12},\be_{34},\be_{56} & \mathfrak{u}(2)\oplus\mathfrak{u}(1) \\ \hline
	\rho>\lambda>\mu &	 \be_{12},\be_{34},\be_{56} & 3\mathfrak{u}(1) \\ \hline
 \end{array}
 \end{align*}
}
\end{enumerate}
\end{Proposition}
\begin{proof}
Let $\varphi\in\Lambda^4 V$ be a rank $6$ fourvector with Euclidean support. 
 Then $g\cdot\varphi$ has support $E=\langle \be_1,\ldots,\be_6\rangle$
for some $g\in G^\theta$, and two fourvectors with support $E$ are in the same orbit for $G^\theta$ if and only if they are in the same orbit for the orthogonal group $\mathrm{O}(E)$ (because the support is preserved, as well as its orthogonal subspace). Therefore it is enough to describe the orbits
of fourvectors $\varphi\in\Lambda^4 E$ of rank $6$ under the action of $\mathrm{O}(E)$.

Fix a volume element in $E^*$ and consider the isomorphism of $\mathrm{SO}(E)$-modules $\Lambda^4 E\cong\Lambda^2 E^*$ 
 given by the contraction.  
Now $\Lambda^2 E^*\cong\mathfrak{so}(E)$,
so the $\mathrm{SO}(E)$-orbits on $\Lambda^4 E$ are in bijective correspondence with adjoint orbits. 
The group $\mathrm{SO}(E)$ is compact and connected, hence any 
adjoint orbit has a representative in the Cartan subalgebra 
$$
\mathfrak{t}=\left\{\left(\begin{array}{cc|cc|cc}
0 & \mu & & & & \\
-\mu & 0 & & &  &\\
\hline
& & 0 & \lambda & &  \\
& &-\lambda & 0 & &  \\
\hline
& & & & 0 & \rho   \\
& & & & -\rho & 0 
\end{array}
\right)\mid\rho,\lambda,\mu\in\mathbb R
\right\}$$ 
of $\mathfrak{so}(E)$. Using the Weyl group of $\mathrm{SO}(E)$, which consists of all permutations and even number of sign changes, we may uniquely arrange for $\rho\geq \lambda\geq \mu$, with $\rho,\lambda\geq 0$. 
If $\mu\geq 0$ the corresponding fourvector is $\varphi_{(\rho,\lambda,\mu,+)}$, if $\mu<0$ we may perform an additional action of the Weyl group to arrive at $\varphi_{(\rho,\lambda,\mu,-)}$. Using an even number of sign changes from Weyl group, we see that in $\Lambda^4V$ orbits for $\mathrm{O}(E)$ coincide with orbits for $\mathrm{SO}(E)$, so claim (1) is settled. Claim (2) follows from straightforward computations, which we omit. 
\end{proof}
\begin{Remark}
Note that $\mathfrak{stab}_{\mathfrak{so}(E)}(\varphi)\subset\mathfrak{u}(3)$ in all cases of Proposition \ref{lem:rank6fourvectors-orthogonal}, with the compatible complex structure on $E$ given by  $J_E:=\be_{12}+\be_{34}+\be_{56}$.
\end{Remark}
\subsection{Reduction to $\mathrm{SO}(V)$-orbits in $\Lambda^4 V$ of rank $6$}
Filtered subdeformations associated to decomposable fourvectors have already been studied in \cite{FOFS2017}. In case of Euclidean support, we thus have the following:

\begin{Theorem}
\label{prop:sugra-decomposable}
Let $(M,g,F)$ be a supergravity background with $4$-form $F$ of rank $\leq 4$ and Euclidean support.
If the space $\mathfrak{k}_{\bar 1}$ of Killing spinors has dimension $\dim\mathfrak{k}_{\bar 1}> 16$, then 
$(M,g,F)$ is locally isometric to the maximally supersymmetric Minkowski spacetime or Freund–Rubin background $AdS_7\times S^4$.
\end{Theorem}
\begin{proof}
Let $\varphi=\rho \be_{1234}$ for some $\rho\geq 0$ and $(M,g,F)$ be a highly supersymmetric supergravity  background with $symb(M,g,F)=(\varphi,S')$ for $S'\subset S$. Since the transvection superalgebra of maximally supersymmetric Minkowski spacetime or Freund–Rubin background is the filtered subdeformation with symbol $(\varphi,S)$ \cite{FOFS2017}, then $S'=S$ by maximality. The claim then follows from Theorem \ref{thm:strong-version-reconstruction}.
\end{proof}
The next step in the analysis of supergravity backgrounds is thus given by the fourvectors of rank $6$.
It is well-known that the field equations    \eqref{eq:bosfieldeqs} are invariant under
a homothety that rescales both the metric and the $4$-form and that the associated transvection superalgebras
are not isomorphic as filtered subdeformations. However, they are isomorphic if we simply allow for $g\in\CSpin(V)$ in \eqref{eq:generaliso}; w.l.o.g. we may thus restrict our analysis from now on to the
fourvectors 
$\varphi_{(\rho,\lambda,\mu)}:=\varphi_{(\rho,\lambda,\mu,+)}=\rho\be_{1234}+\lambda\be_{1256}+\mu\be_{3456}$  (the $\rho,\lambda,\mu$ as in Proposition \ref{lem:rank6fourvectors-orthogonal}).

\section{The interplay of spinors and fourvectors}
\label{sec:4}
\subsection{Clifford algebras and admissible bilinear forms}
We set up a spinorial formalism that is adapted to $\varphi=\varphi_{(\rho,\lambda,\mu)}$, for all values of $\rho,\lambda,\mu$ as in Proposition \ref{lem:rank6fourvectors-orthogonal} at the same time. We consider the support $E=\langle \be_1,\ldots,\be_6\rangle$ of $\varphi$, set $F=\langle \be_7,\be_8,\be_9,\be_\natural\rangle$ and $E^\perp=F\oplus\langle\be_0\rangle$. Then 
$V=E\oplus E^\perp=E\oplus F\oplus\langle\be_0\rangle$ as orthogonal direct sums. For our purposes, it will often be convenient to work at the complexified level, we thus write $\mathbb V:=V\otimes\mathbb C$, $\mathbb E:=E\otimes\mathbb C$, etc.

 We recall that the Clifford algebras $\Cl(E)\cong \mathbb H(4)$ and 
$\Cl(F)\cong\mathbb H(2)$, acting irreducibly on the respective spinor modules $\Sigma\cong \mathbb H^4$ and $\Delta\cong \mathbb H^2$. Moreover
$$
\begin{aligned}
\Sigma&\cong \Sigma_+\oplus\Sigma_-\;,\\
\Delta&\cong\Delta_+\oplus\Delta_-\;,
\end{aligned}
$$
as representations for $\mathfrak{so}(E)$ and $\mathfrak{so}(F)$, respectively. It is known that the semispinor modules 
$\Sigma_{\pm}\cong \mathbb C^4$ are isomorphic, whereas $\Delta_\pm\cong \mathbb H$ are not. In particular, the Schur algebras $\mathcal C(\Sigma)$ of $\Sigma$ and $\mathcal C(\Delta)$ of $\Delta$ (the algebra of the endomorphisms that are invariant for the respective orthogonal Lie algebras) are isomorphic to $\mathbb C(2)$ and $\mathbb H\oplus\mathbb H$.

We will freely use the results and notations of \cite{AC1}. In particular, it is shown there that each Schur algebra admits a basis of ``admissible'' endomorphisms (an endomorphism is admissible if it has three invariants attached, called $(\tau,\sigma,\imath)\in\mathbb Z_2^3$: the first says whether the endomorphism commutes or anticommutes with Clifford multiplication by vectors, the third whether it preserves or it exchanges the semispinor modules, the second denotes its symmetry w.r.t. a canonically-defined bilinear form). The following tables recollect the basis elements of the even part of the two Schur algebras (namely the elements with $\tau=+1$). 
The notation of the last two basis elements of $C(\Delta)_{\bar 0}$ does not match the one in \cite{AC1}, in order to avoid some unpleasant overlap with the notations used in this paper.

{\footnotesize
\begin{align*}
 \begin{array}{|c|c|c|c|} \hline
\text{Basis elements of $\mathcal C(\Sigma)_{\bar 0}$}  &\text{Invariants $(\tau,\sigma,\imath)$} & \text{Basis elements of $\mathcal C(\Delta)_{\bar 0}$} & \text{Invariants $(\tau,\sigma,\imath)$}\\ \hline\hline
\Id & (+++) & \Id & (+++)\\ \hline
  I &(+-+) & I & (+-+) \\ \hline
	J & (+--) & 
	J & (+-+) \\ \hline
	K &	 (+--) & 
	K & (+-+)\\ \hline
 \end{array}
 \end{align*}
}
\vskip0.2cm\par\noindent
The volume elements $\vol_E$ and $\vol_F$ in $E$ and $F$ have invariants $(\tau,\sigma,\imath)=(--+)$ and $(\tau,\sigma,\imath)=(-++)$, respectively, and satisfy $\vol_E^2=-1$, $\vol_F^2=+1$. In fact, $\vol_F$ acts semisimply on $\Delta$ with eigenvalues $\pm 1$ on $\Delta_\pm$, and similarly $-\vol_E I$ on $\Sigma_\pm$. Finally, we will often regard $\Sigma$ and $\Delta$ as
$8$-dimensional and $4$-dimensional complex vector spaces thanks to the complex structure $I$. They coincide with the complex (Dirac, in the physics terminology) spinor modules in the respective dimensions.
\begin{Lemma}
\label{lemma:action}
\hfill
\begin{enumerate}
\item $\Cl(E\oplus F)\cong \Cl(E)\otimes_{\mathbb R}\Cl(F)$;
\item $\Cl(E\oplus F)$ acts on $\Sigma\otimes_{\mathbb C}\Delta$ preserving the conjugation $c:=J\otimes J$;
\item $c$ exchanges $\Sigma_+\otimes_{\mathbb{C}}\Delta_\pm$ and $\Sigma_-\otimes_{\mathbb{C}}\Delta_\pm$; 
\item The above action of $\Cl(E\oplus F)$ on $\Sigma\otimes_{\mathbb C}\Delta$ extends uniquely to an action of $\Cl(V)$ in such a way that $\vol_V\cdot s=-s$ for all $s\in S$ and $c$ is still an invariant conjugation.
\end{enumerate}
\end{Lemma}
\begin{proof}
Consider the linear map $E\oplus F\to \Cl(E)\otimes_{\mathbb R}\Cl(F)$ sending $v\in E$ to $v\otimes \vol_F$ and $v\in F$ to $1\otimes v$. Since $\vol_F^2=+1$, this extends to an algebra epimomorphism from $\Cl(E\oplus F)$ to $\Cl(E)\otimes_{\mathbb R}\Cl(F)$, which is injective by dimensional reasons. This proves (1), and $(2)$ is then clear, since each $J$ is a quaternionic structure on a complex vector space with invariant $\tau=+1$. Claim (3) is immediate from the invariant $\imath$. 

The last claim follows straightforwardly since $\Cl(E\oplus F)\cong \mathbb R(32)$ and $\Cl(V)\cong 2\mathbb R(32)$. (More explicitly, since $\vol_V=\be_0\vol_E\vol_F$ and $\be_0^2=-1$, it is enough to declare the action of $\be_0
$
on $\Sigma\otimes_{\mathbb C}\Delta$ to be equal to $\vol_E\otimes \vol_F$.)
\end{proof}
The $32$-dimensional complex vector space $\mathbb S:=\Sigma\otimes_{\mathbb C}\Delta$ with the above action of $\Cl(V)$ is our model for the complexification of $S$, which is the fixed point set of $c$.  In practice we will work with $\mathbb S=S\otimes_{\mathbb R} \mathbb C$, although we will not always mention this. By \cite{AC1}, there exists a non-degenerate complex bilinear form $f:\Sigma\otimes_{\mathbb C}\Sigma\to\mathbb C$ with invariants $(\tau,\sigma,\imath)=(-+-)$, i.e., Clifford multiplication by a vector is skew-symmetric, $f$ is symmetric, and $\Sigma_\pm$ is isotropic. It is unique up to scalars. A similar result is true for 
$\Delta$ with the invariants $(\tau,\sigma,\imath)=(--+)$, and we also denote this form by $f$, with a little abuse of notation.
\begin{Lemma}
The bilinear form $\langle-,-\rangle:=f\otimes f$ on the $\Cl(V)$-module $\mathbb S$ has invariants $(\tau,\sigma)=(-1,-1)$ and it can be arranged so that its pull-back via $c$ is equal to its conjugate. More precisely, each $f$ can be arranged so that its pull-back via $J$ is equal to its conjugate.
\end{Lemma}
\begin{proof}
The claims on invariants are all immediate, except perhaps $\tau=-1$ for any $v\in E$ (this follows since $v$ acts as $v\otimes \vol_F$ and $f(t_1,\vol_F t_2)=f(\vol_F t_1,t_2)$ for all $t_1,t_2\in \Delta$). 
A multiple of $\langle-,-\rangle$ coincides necessarily with the complexification of the unique skew-symmetric invariant bilinear form on $S$, so we may rescale it appropriately to ensure that 
$c^*\langle-,-\rangle=\overline{\langle-,-\rangle}$.
It follows that $J^*f=c\overline{f}$ for some non-zero $c\in\mathbb C$, so that
\begin{equation*}
\begin{aligned}
f&=J^*(J^*f)=J^*(c\overline{f})=\overline{c}J^*\overline{f}=\overline{c}\overline{J^*f}=\overline{c}^2f	\\ &\Longrightarrow
c=\pm 1\;.
\end{aligned}
\end{equation*}
If $c=+1$ we are done, otherwise it is sufficient to multiply $f$ with either $+i$ or $-i$.
\end{proof}
Let $f$ with a subscript be the composition of $f$ with an endomorphism in the second entry.
\begin{Corollary}
\label{cor:hermitianforms}
The restrictions of the forms $if_J$ to $\Sigma_\pm$ and $f_J$ to $\Delta_\pm$ are definite Hermitian forms.
\end{Corollary}
\begin{proof}
For all $\upsilon_1,\upsilon_2\in\Sigma_\pm$, we have
\begin{equation*}
\begin{aligned}
\overline{if_J(\upsilon_1,\upsilon_2)}&=-i\overline{f(\upsilon_1,J\upsilon_2)}=if(J\upsilon_1,\upsilon_2)
=if(\upsilon_2,J\upsilon_1)\\
&=if_J(\upsilon_2,\upsilon_1)\;,
\end{aligned}
\end{equation*}
so $if_J$ is pseudo-Hermitian on $\Sigma_\pm$. It is invariant for the action of $\mathrm{Spin}(E)\cong \mathrm{SU}(4)$, therefore it is definite separately on each $\Sigma_\pm$. The proof for $f_J$ on $\Delta_\pm$ is similar, we omit it.
\end{proof}
\begin{Remark}
Using an explicit realization of the Clifford algebras, one can check that each form of Corollary \ref{cor:hermitianforms} has overall split signature, but we won't need this fact.
For concreteness, we will assume they are positive definite restricted to $\Sigma_+$ and to $\Delta_+$, respectively.
\end{Remark}
\subsection{The Dirac current}

We here provide a qualitative analysis of the (complexified) Dirac current $\kappa:\odot^2\mathbb S\to \mathbb V$, under the decomposition 
$\odot^2 \mathbb S=\odot^2\big(\Sigma\otimes_{\mathbb C}\Delta\big)=\odot^2\Sigma\,\odot^2\Delta\;\oplus\;\Lambda^2\Sigma\;\Lambda^2\Delta$. Using the defining equation 
of $\kappa$ and the action of $\Cl(V)$ on $\mathbb S$ constructed in Lemma \ref{lemma:action}, we see that 
\begin{equation}
\label{eq:Diraccurrentbroken}
\begin{aligned}
\kappa= \underbrace{\kappa_{\mathbb E}\otimes f_{\vol_F}}_{\text{values in}\;\mathbb E}+\underbrace{f\otimes\kappa_{\mathbb F}}_{\text{values in}\;\mathbb F}+\underbrace{\big(f_{\vol_E}\otimes f_{\vol_F}\big)\be_0}_{\text{values in}\;\mathbb C\be_0}\;,
\end{aligned}
\end{equation}
where $\kappa_{\mathbb E}$ (resp. $\kappa_{\mathbb F}$) denotes the Dirac current operator on $\Sigma$  (resp. $\Delta$) constructed using $f$.
From a purely representation-theoretic point of view, $\kappa_{\mathbb E}:\Lambda^2\Sigma\to\mathbb E$ induces an isomorphism $\Lambda^2\Sigma_{\pm}\cong\mathbb E$ while $\kappa_{\mathbb F}:\odot^2\Delta\to\mathbb F$ induces an isomorphism $\Delta_+\odot\Delta_-\cong\mathbb F$. 
It is straightforward to see that the invariants of $f_{\vol_E}$ and $f_{\vol_F}$ are $(\tau,\sigma,\imath)=(+--)$ and $(\tau,\sigma,\imath)=(+-+)$, respectively, in particular they  are both skew-symmetric.

\begin{Corollary}
$(\Sigma_+\otimes\Delta_+)\odot(\Sigma_+\otimes\Delta_-)$ and its conjugate $(\Sigma_-\otimes\Delta_+)\odot(\Sigma_-\otimes\Delta_-)$ are included in the kernel of the Direc current $\kappa$.
\end{Corollary}
\begin{proof}
It follows from \eqref{eq:Diraccurrentbroken} and the invariants $\imath=+1$ of $f_{\vol_F}$ and $\imath=-1$ of $f$ on $\Sigma$.
\end{proof}
\subsection{Curvature in superspace}
We study the behaviour of map $\gamma^\varphi:\odot^2 S\to\fso(V)\cong\Lambda^2 V$ associated to the fourvectors $\varphi=\varphi_{(\rho,\lambda,\mu)}\in\Lambda^4 V$, in particular its dependence on varying the parameters $\rho\geq \lambda\geq \mu\geq 0$ (only $\mu$ can be zero). This is a crucial step in view of the proofs of the main Theorem \ref{mainthm:1} and Theorem \ref{mainthm:2} in the forthcoming \S\ref{sec:5} and \S\ref{sec:6}.

Since 
$\odot^2 S\cong\Lambda^1 V\oplus \Lambda^2 V\oplus\Lambda^5 V$ and $\Ker\gamma^\varphi\supset\Lambda^1 V$, we get $\operatorname{Im}(\gamma^\varphi)=\gamma^\varphi(\Lambda^2V)+\gamma^\varphi(\Lambda^5V)$.
We substitute the identity $\star\varphi=\be_0\wedge \vol_F\wedge\bigstar\varphi$ in \eqref{definition-gamma-useful}, where $\bigstar$ is the Hodge star operator on $E$ and $\bigstar\varphi=\mu\be_{12}+\lambda\be_{34}+\rho\be_{56}$, and by straightforward computations we arrive at:
\begin{Proposition}
\label{prop:imgamma}
The following diagram subsumes the action of $\gamma^\varphi$ on the components of the spaces $\Lambda^2V$ and $\Lambda^5V$ w.r.t. the decomposition $V=E\oplus F\oplus\langle\be_0\rangle$.
\vskip0.15cm\par\noindent
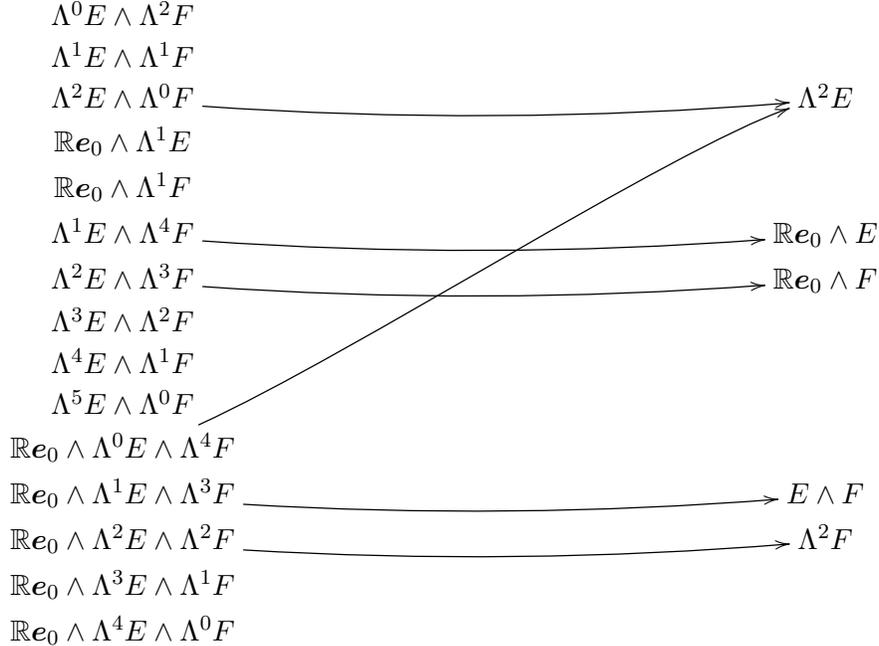
$$
{\small
\xymatrix@C=1pc@R=0pc{
\Lambda^0E\wedge\Lambda^2F& & & & & & & & & & &\\
\Lambda^1E\wedge\Lambda^1F& & & & & & & & & & &\\
\Lambda^2E\wedge\Lambda^0F\ar@/_/[rrrrrrrrrrr]& & & & & & & & & & &\Lambda^2E\\
\mathbb R\be_0\wedge \Lambda^1E& & & & & & & & & & &\\
\mathbb R\be_0\wedge \Lambda^1F& & & & & & & & & & &\\
\Lambda^1 E\wedge\Lambda^4F \ar@/_/[rrrrrrrrrrr]& & & & & & & & & & & \mathbb R\be_0\wedge E\\
\Lambda^2 E\wedge\Lambda^3F\ar@/_/[rrrrrrrrrrr]& & & & & & & & & & & \mathbb R\be_0\wedge F\\
\Lambda^3 E\wedge\Lambda^2F& & & & & & & & & & & \\
\Lambda^4 E\wedge\Lambda^1F& & & & & & & & & & & \\
\Lambda^5 E\wedge\Lambda^0F& & & & & & & & & & &\\
\mathbb R\be_0\wedge\Lambda^0E\wedge\Lambda^4F\ar@(r,l)[uuuuuuuurrrrrrrrrrr] & & & & & & & & & & &\\
\mathbb R\be_0\wedge\Lambda^1E\wedge\Lambda^3F\ar@/_/[rrrrrrrrrrr]& & & & & & & & & & &E\wedge F\\
\mathbb R\be_0\wedge\Lambda^2E\wedge\Lambda^2F\ar@/_/[rrrrrrrrrrr] & & & & & & & & & & &\Lambda^2F \\
\mathbb R\be_0\wedge\Lambda^3E\wedge\Lambda^1F & & & & & & & & & & & \\
\mathbb R\be_0\wedge\Lambda^4E\wedge\Lambda^0F & & & & & & & & & & & \\
}
}
$$
\vskip0.15cm\par\noindent
More precisely:
\begin{enumerate}
\item The first arrow is a bijection, unless $\mu=0$, in which case
\begin{equation*}
\begin{aligned}
\!\!\!\!\!\operatorname{Im}(\gamma^\varphi|_{\Lambda^2E\wedge\Lambda^0F})&=\langle \be_{13},\be_{14},\be_{15},\be_{16},\be_{23},\be_{24},\be_{25},\be_{26},\be_{12},
\rho\be_{34}+\lambda\be_{56}
\rangle\;,\;\;\;\;\;\;\;\;\;\;\;\;\;\;\;\;\;\;\;\;\;\;\;\;\;\;\;\;\;\;\\
\operatorname{Ker}(\gamma^\varphi|_{\Lambda^2E\wedge\Lambda^0F})&=\langle\be_{35}, \be_{36},\be_{45},\be_{46},\lambda\be_{34}-\rho\be_{56}\rangle
\;;
\end{aligned}
\end{equation*}
\item The second arrow is a bijection, unless $\mu=0$, in which case 
\begin{equation*}
\begin{aligned}
\operatorname{Im}(\gamma^\varphi|_{\Lambda^1E\wedge\Lambda^4F})&=\mathbb R\be_0\wedge\langle\be_3,\ldots,\be_6\rangle\;,\\
\operatorname{Ker}(\gamma^\varphi|_{\Lambda^1E\wedge\Lambda^4F})&=\langle\be_1,\be_2\rangle\wedge\vol_F\;;
\end{aligned}
\end{equation*}
\item The third arrow is surjective with kernel $(\bigstar\varphi)^\perp\wedge\Lambda^3F$;
\item The fourth arrow is injective with image generated by $\bigstar\varphi$; 
\item The fifth arrow is a bijection, unless $\mu=0$, in which case 
\begin{equation*}
\begin{aligned}
\operatorname{Im}(\gamma^\varphi|_{\mathbb R\be_0\wedge\Lambda^1E\wedge\Lambda^3F})&=\langle \be_3,\ldots,\be_6\rangle\wedge F\;,\\
\operatorname{Ker}(\gamma^\varphi|_{\mathbb R\be_0\wedge\Lambda^1E\wedge\Lambda^3F})&=\mathbb R\be_0\wedge\langle\be_1,\be_2\rangle\wedge\Lambda^3F\;;
\end{aligned}
\end{equation*}
\item The sixth arrow is surjective with kernel $\mathbb R \be_0\wedge(\bigstar\varphi)^\perp\wedge\Lambda^2F$; 
\end{enumerate} 
where $(\bigstar\varphi)^\perp$ is the orthogonal to $\bigstar\varphi$ in $\Lambda^2E$.
\end{Proposition}
From Proposition \ref{lem:rank6fourvectors-orthogonal} and Proposition \ref{prop:imgamma}, we immediately get the following.
\begin{Corollary}
The Lie algebra $\operatorname{Im}(\gamma^\varphi)\cap \stab_{\mathfrak{so}(V)}(\varphi)$ coincides always with $\stab_{\mathfrak{so}(V)}(\varphi)$, unless $\mu=0$, $\lambda=\rho$, in which case it coincides with $\mathfrak{so}(E^\perp)\oplus\langle \be_{12},\be_{34}+\be_{56}\rangle$.
\end{Corollary}
We focus on the nodes of the diagram of Proposition \ref{prop:imgamma} from which an arrow emanates and work at the complexified level. Using the defining equations of $\omega^{(2)}$, $\omega^{(5)}$ and the action of $\Cl(V)$ on $\mathbb S$ constructed in Lemma \ref{lemma:action}, we have for all $\omega\in\odot^2\mathbb S$ that
\begin{equation}
\label{eq:higherDiraccurrentbroken}
\begin{aligned}
\omega^{(2)}&\equiv \underbrace{\omega^{(2)}_{\mathbb E}\otimes f}_{\text{values in}\;\Lambda^2 \mathbb E\wedge\Lambda^0 \mathbb F}\mod \operatorname{Ker}(\gamma^\varphi)\cap\Lambda^2 \mathbb V\;,\\
\omega^{(5)}&\equiv \underbrace{\big(\kappa_{\mathbb E}\otimes f\big)\vol_F}_{\text{values in}\;\Lambda^1 \mathbb E\wedge\Lambda^4 \mathbb F}
+\underbrace{\omega^{(2)}_{\mathbb E}\otimes\omega^{(3)}_{\mathbb F}}_{\text{values in}\;\Lambda^2\mathbb E\wedge\Lambda^3\mathbb F}\\
&\;\;\;+\underbrace{\big(f_{\vol_E}\otimes f\big)\be_0\wedge\vol_F}_{\text{values in}\;\mathbb C\be_0\wedge\Lambda^0\mathbb E\wedge\Lambda^4\mathbb F}+\underbrace{\big(\widetilde\kappa_{\mathbb E}\otimes\omega^{(3)}_{\mathbb F}\big)\be_0}_{\text{values in}\;\mathbb C\be_0\wedge\Lambda^1 \mathbb E\wedge\Lambda^3\mathbb F}\\
&\;\;\;+\underbrace{\big(\widetilde\omega^{(2)}_{\mathbb E}\otimes \widetilde\omega^{(2)}_{\mathbb F}\big)\be_0}_{\text{values in}\;\mathbb C\be_0\wedge\Lambda^2\mathbb E\wedge\Lambda^2\mathbb F}\mod \operatorname{Ker}(\gamma^\varphi)\cap\Lambda^5\mathbb V\;,
\end{aligned}
\end{equation}
with 
\begin{enumerate}
	\item $\omega^{(2)}_{\mathbb E}:\Sigma_+\wedge\Sigma_-\to\Lambda^2\mathbb E$ the polyvector operator on $\Sigma$ constructed via $f$;
	\item $\omega^{(3)}_{\mathbb F}:\Delta_+\wedge\Delta_-\to\Lambda^3\mathbb F$ the polyvector operator on $\Delta$ constructed via $f$; 
	\item $\widetilde\kappa_{\mathbb E}:\Lambda^2\Sigma_\pm\to\mathbb E$ the Dirac current on $\Sigma$ constructed via $f_{\vol_{E}}$; 
	\item $\widetilde\omega^{(2)}_{\mathbb E}:\Sigma_+\odot\Sigma_-\to\Lambda^2\mathbb E$ the polyvector operator on $\Sigma$ constructed via $f_{\vol_{E}}$;
	\item $\widetilde\omega^{(2)}_{\mathbb F}:\odot^2\Delta_\pm\to\Lambda^2_\pm\mathbb F$ the polyvector operator on $\Delta$ constructed via $f_{\vol_{F}}$. 
\end{enumerate}
In particular the contributions in \eqref{eq:higherDiraccurrentbroken} all act on $\Lambda^2\Sigma\;\Lambda^2\Delta$, except the last acting on $\odot^2\Sigma\,\odot^2\Delta$. The maps
$\omega^{(3)}_{\mathbb F}$, $\widetilde\kappa_{\mathbb E}$, $\widetilde\omega^{(2)}_{\mathbb F}$ are isomorphisms, $\omega^{(2)}_{\mathbb E}$ and $\widetilde\omega^{(2)}_{\mathbb E}$ surjective with $1$-dimensional kernel. In particular, it is checked directly from the definitions that, in both cases, the kernel does not include any decomposable tensor of $\Sigma_+\otimes\Sigma_-$, by irreducibility of $\Sigma_\pm$ as  $\mathfrak{so}(\mathbb E)$-modules.
\vskip0.3cm\par
We may now combine \eqref{eq:higherDiraccurrentbroken} with Proposition \ref{prop:imgamma}, arriving at the complete diagram for the action of $\gamma^\varphi$. Since this is a crucial ingredient for our following arguments, together with the action of the Dirac current, we will also report arrows for the latter. 
\begin{Proposition}
\label{prop:imgammauseful}
The following diagram subsumes the action of the maps $\kappa$ and $\gamma^\varphi$ (resp. on the left and on the right of the diagram) on the complexification $\odot^2 \mathbb S$ w.r.t. the decomposition
$\odot^2 \mathbb S=\odot^2\Sigma\,\odot^2\Delta\;\oplus\;\Lambda^2\Sigma\;\Lambda^2\Delta$ and the decompositions of $\Sigma$, $\Delta$ into semispinor modules.
$$
{\small
\xymatrix@C=1pc@R=0pc{
&&&&&&(\Sigma_+\odot\Sigma_+)\otimes(\Delta_+\odot\Delta_+)& & & & & &  \\
&&&&&&(\Sigma_+\odot\Sigma_+)\otimes(\Delta_-\odot\Delta_-)& & & & & &  \\
&&&&&&(\Sigma_+\odot\Sigma_+)\otimes(\Delta_+\odot\Delta_-)& & & & & &  \\
&&&&&&(\Sigma_-\odot\Sigma_-)\otimes(\Delta_+\odot\Delta_+)& & & & & &  \\
&&&&&&(\Sigma_-\odot\Sigma_-)\otimes(\Delta_-\odot\Delta_-)& & & & & &  \\
&&&&&&(\Sigma_-\odot\Sigma_-)\otimes(\Delta_+\odot\Delta_-)& & & & & &  \\
&&&&&&(\Sigma_+\odot\Sigma_-)\otimes(\Delta_+\odot\Delta_+)\ar@/_/[rrrrrr]& & & & & & \Lambda^2\mathbb F\\
&&&&&&(\Sigma_+\odot\Sigma_-)\otimes(\Delta_-\odot\Delta_-)\ar@/_/[urrrrrr]& & & & & & \\
\mathbb F&&&&&&\ar@(dl,u)[llllll](\Sigma_+\odot\Sigma_-)\otimes(\Delta_+\odot\Delta_-)& & & & & & \\
\mathbb E&&&&&&\ar@(dl,u)[llllll](\Sigma_+\wedge\Sigma_+)\otimes(\Delta_+\wedge\Delta_+)\ar@/_/[rrrrrr]& & & & & & \mathbb C\be_0\wedge \mathbb E\\
&&&&&&\ar@(dl,u)[ullllll](\Sigma_+\wedge\Sigma_+)\otimes(\Delta_-\wedge\Delta_-)\ar@/_/[urrrrrr]& & & & & & \\
&&&&&&(\Sigma_+\wedge\Sigma_+)\otimes(\Delta_+\wedge\Delta_-)\ar@(r,l)[rrrrrr]& & & & & & \mathbb E\wedge\mathbb F\\
&&&&&&\ar@(dl,u)[uuullllll](\Sigma_-\wedge\Sigma_-)\otimes(\Delta_+\wedge\Delta_+)\ar@(r,d)[uuurrrrrr]& & & & & &  \\
&&&&&&\ar@(dl,u)[uuuullllll](\Sigma_-\wedge\Sigma_-)\otimes(\Delta_-\wedge\Delta_-)\ar@(r,d)[uuuurrrrrr]& & & & & & \\
&&&&&&(\Sigma_-\wedge\Sigma_-)\otimes(\Delta_+\wedge\Delta_-)\ar@(r,d)[uuurrrrrr]& & & & & & \\
\mathbb C\be_0&&&&&&\ar@(r,l)[llllll](\Sigma_+\wedge\Sigma_-)\otimes(\Delta_+\wedge\Delta_+)\ar@/_/[rrrrrr]& & & & & & \Lambda^2 \mathbb E\\
&&&&&&\ar@(dl,u)[ullllll](\Sigma_+\wedge\Sigma_-)\otimes(\Delta_-\wedge\Delta_-)\ar@/_/[urrrrrr]& & & & & & \\
&&&&&&(\Sigma_+\wedge\Sigma_-)\otimes(\Delta_+\wedge\Delta_-)\ar@(l,r)[rrrrrr]& & & & & & \mathbb C\be_0\wedge \mathbb F
}
}
$$
\end{Proposition}
\vskip0.3cm\par

We will now turn to using the presentation of $S$ adapted to $\varphi=\varphi_{(\rho,\lambda,\mu)}\in\Lambda^4 V$ and the results of \S\ref{sec:4} to obtain a priori estimates of the size of the Dirac kernel $\mathfrak D$ of $S'$, in a manner that depends only on the dimension of $S'$. The idea is then very simple: we will exploit that $\mathfrak D$ is sufficiently big to contradict the quite restrictive condition $\gamma^\varphi(\mathfrak D)=\mathfrak h_{(\varphi,S')}\subset\mathfrak{stab}_{\fso(V)}(\varphi)$ (recall that the notion
of a Lie pair is given in Definition  \ref{def:Lie pair} and that Proposition   \ref{prop:LiePair} holds). We consider the case where $\varphi=\varphi_{(\rho,\lambda,\mu)}$ has length $3$ first in \S\ref{sec:5} and then turn to the case of lenght $2$ in \S\ref{sec:6}, which is more involved.
\section{The main theorem: case of fourvectors of rank $6$ and length $3$}
\label{sec:5}
We here establish the following result.
\begin{Theorem}
\label{mainthm:1}
Let $(M,g,F)$ be a supergravity  background with 
$4$-form $F$ of rank $6$, lenght $3$, and the Euclidean support. Then the space $\mathfrak{k}_{\bar 1}$ of Killing spinors has dimension $\dim \mathfrak{k}_{\bar 1}\leq 26$.
\end{Theorem}
\begin{proof}
If $(M,g,F)$ is not highly supersymmetric, the claim is obvious. Let $(M,g,F)$ be highly supersymmetric and
let $symb(M,g,F)=(\varphi,S')$ be the corresponding geometric symbol, where 
$\varphi=\varphi_{(\rho,\lambda,\mu)}$ with $\mu\neq 0$ by assumptions and $S'\subset S$. We set $\mathbb S':=S'\otimes_{\mathbb R}\mathbb C\subset \mathbb S=\Sigma\otimes_{\mathbb C}\Delta$.
\vskip0.2cm\par
If $(\Sigma_+\otimes\Delta_+)\cap\mathbb S'=(0)$, then we may consider the natural projection
\begin{equation}
\label{eq:projection}
\pi:\mathbb S\rightarrow\mathbb S/(\Sigma_+\otimes\Delta_+)\;,
\end{equation}
which is injective restricted to $\mathbb S'$, thus $\dim\mathbb S'\leq 24$. A completely similar argument holds if $(\Sigma_+\otimes\Delta_-)\cap\mathbb S'=(0)$. (The other two cases are automatic by conjugation, see Lemma \ref{lemma:action}.)
\vskip0.2cm\par
Let us now assume that both $(\Sigma_+\otimes\Delta_+)\cap\mathbb S'$ and $(\Sigma_+\otimes\Delta_-)\cap\mathbb S'$ are non-trivial. We fix a basis $\left\{\epsilon_1^\pm,\epsilon_2^\pm\right\}$ of $\Delta_{\pm}$ and consider
two non-zero elements
\begin{equation}
\label{eq:first-step-elements}
\begin{aligned}
s^+&=\sigma_1\otimes\epsilon_1^+ +\sigma_2\otimes\epsilon_2^+\in(\Sigma_+\otimes\Delta_+)\cap\mathbb S'\;,\\
s^-&=\varsigma_1\otimes\epsilon_1^-+\varsigma_2\otimes\epsilon_2^-\in(\Sigma_+\otimes\Delta_-)\cap\mathbb S'\;,
\end{aligned}
\end{equation}
where $\sigma_1,\ldots,\varsigma_2\in\Sigma_{+}$. Now $\kappa(s^+,s^-)=0$ by Proposition \ref{prop:imgammauseful}, so $s^+\odot s^-\in\mathfrak D\otimes_{\mathbb R}\mathbb C$. Moreover, combining \eqref{eq:higherDiraccurrentbroken} with the right side of the diagram of Proposition \ref{prop:imgammauseful},we get
\begin{equation}
\label{eq:crucialevaluation}
\gamma^\varphi(s^+,s^-)=\gamma^\varphi\big\{\underbrace{\big(\widetilde\kappa_{\mathbb E}\otimes\omega^{(3)}_{\mathbb F}(s^+,s^-)\big)\be_0}_{\text{element of}\;\mathbb C\be_0\wedge\Lambda^1\mathbb E\wedge\Lambda^3\mathbb F}\big\}\;,
\end{equation}
where $\widetilde\kappa_{\mathbb E}:\Lambda^2\Sigma_+\to\mathbb E$ and $\omega^{(3)}_{\mathbb F}:\Delta_+\wedge\Delta_-\to\Lambda^3\mathbb F$ are isomorphisms. Since the restriction of $\gamma^\varphi$ to 
$\mathbb C\be_0\wedge\Lambda^1\mathbb E\wedge\Lambda^3\mathbb F$ is an isomorphism with image $\mathbb E\wedge\mathbb F\subset\mathfrak{so}(\mathbb V)$ by Proposition \ref{prop:imgamma} (here we are using the assumption that $\mu\neq 0$), we see that \eqref{eq:crucialevaluation} vanishes if and only 
$$\widetilde\kappa_{\mathbb E}(\sigma_1,\varsigma_1)=\widetilde\kappa_{\mathbb E}(\sigma_1,\varsigma_2)=\widetilde\kappa_{\mathbb E}(\sigma_2,\varsigma_1)=\widetilde\kappa_{\mathbb E}(\sigma_2,\varsigma_2)=0\;.$$
In other words, if and only if $\sigma_1,\ldots,\varsigma_2$ lie on the same complex line in $\Sigma_+$, i.e., 
there exists a $\sigma\in\Sigma_+$ such that $s^+\in \mathbb C\sigma\otimes\Delta_+$ and $s^-\in\mathbb C\sigma\otimes\Delta_-$.

We first claim that this is necessarily the case. Otherwise, \eqref{eq:crucialevaluation} would not vanish and $\gamma^\varphi(\mathfrak D)$ would include a non-zero element in $E\wedge F$, contradicting the identity $\gamma^\varphi(\mathfrak D)\subset \stab_{\mathfrak{so}(V)}(\varphi)$ for highly-supersymmetric backgrounds. By genericity of $s^+$ and $s^-$, we thus infer
\begin{equation*}
\begin{aligned}
(\Sigma_+\otimes\Delta_+)\cap\mathbb S'&\subset\mathbb C\sigma\otimes\Delta_+\;,\\
(\Sigma_+\otimes\Delta_-)\cap\mathbb S'&\subset\mathbb C\sigma\otimes\Delta_-\;,
\end{aligned}
\end{equation*}
and $\dim\mathbb S'\leq 26$ arguing with the projection \eqref{eq:projection} as in the beginning of the proof.
\end{proof}

\section{The main theorem: case of fourvectors of rank $6$ and length $2$}
\label{sec:6}
We turn to the case where the $4$-form $F$ has rank $6$ and Euclidean support, but lenght $2$. 
We shall devote the remainder of the article to prove the following.
\begin{Theorem}
\label{mainthm:2}
Let $(M,g,F)$ be a supergravity  background with 
$4$-form $F$ of rank $6$, lenght $2$, and the Euclidean support. Then the space $\mathfrak{k}_{\bar 1}$ of Killing spinors has dimension $\dim \mathfrak{k}_{\bar 1}\leq 26$.
\end{Theorem}

Again, we may assume that $(M,g,F)$ is highly supersymmetric, otherwise there is nothing to prove.
Then the geometric symbol $symb(M,g,F)=(\varphi,S')$ where 
$\varphi=\varphi_{(\rho,\lambda,\mu)}$ with $\mu=0$. {\it We will then have $\dim S'>16$ and $\mu=0$ tacitly from now.}

We depart with the following simple but very useful result.
\begin{Lemma}
\label{lemma:eigenvectorsbasis}
There exists a complex basis $\left\{\upsilon_+,\upsilon_-,\omega_+,\omega_-\right\}$ of $\Sigma_+$ consisting of eigenvectors for the Clifford action of $\be_{12}$, $\be_{34}$, $\be_{56}$ of $\Lambda^2E\cong \mathfrak{so}(E)$ with eigenvalues as follows:
\vskip0.1cm\par\noindent
{\footnotesize
\begin{align*}
 \begin{array}{|c|c|c|c|c|} \hline
\text{Basis elements of $\Sigma_{+}$}  &\text{Eigenvalues of $\be_{12}$} & \text{Eigenvalues of $\be_{34}$} & \text{Eigenvalues of $\be_{56}$} & \text{Eigenvalues of $\varphi$}\\ \hline\hline
 \upsilon_+& +I & +I & -I &-\rho+\lambda\\ \hline
 \upsilon_- & +I & -I &  +I&\rho-\lambda\\ \hline
 \omega_+& -I & +I &  +I&\rho+\lambda\\ \hline
 \omega_-&	-I & -I &  -I&-\rho-\lambda\\ \hline
 \end{array}
 \end{align*}
}
\vskip0.1cm\par\noindent
Moreover the basis may be assumed orthonormal w.r.t. Hermitian form $if_J$ of Corollary \ref{cor:hermitianforms}.
\end{Lemma}
\begin{proof}
First note that the operators $\be_{12}$, $\be_{34}$, $\be_{56}$,  square to $-\Id$ and pairwise commute.
The operator $\be_{12}$ is thus diagonalizable, and the multiplicities of its eigenvalues $\pm I$ have to coincide (otherwise we get a contradiction: either $\be_{12}$ would be a multiple of the identity or the even part $\mathbb C(2)\oplus\mathbb C(2)$ of the complex Clifford algebra  generated by $\be_3,\ldots,\be_6$ would act on a $3$-dimensional eigenspace, thus trivially on a line).  
The claim then follows because $\be_{12}$, $\be_{34}$, $\be_{56}$ are linearly independent and $\vol_E=I$ on $\Sigma_+$, and then the eigenvalues of $\varphi$ are obtained immediately.

The last claim is immediate from Corollary \ref{cor:hermitianforms}, the equivariancy of $if_J$ w.r.t. $\mathfrak{so}(E)$, and the eigenvalue structure of $\be_{12},\be_{34},\be_{56}$ detailed above.
\end{proof}

Another nice consequence of Lemma \ref{lemma:eigenvectorsbasis} is the following. Setting 
\begin{align*}
\Upsilon&:=\langle\upsilon_+,\upsilon_-\rangle\;,\\
\Omega&:=\langle\omega_+,\omega_-\rangle\;,
\end{align*} we have $\Sigma_+=\Upsilon\oplus\Omega$, $\Sigma_-=J\Upsilon\oplus J\Omega$ with $f(\Upsilon,J\Omega)=f(\Omega,J\Upsilon)=0$. In particular $f$ is a non-degenerate pairing when restricted to $\Upsilon\otimes J\Upsilon$ and, respectively, to $\Omega\otimes J\Omega$.
\vskip0.4cm\par
We will split the proof of Theorem \ref{mainthm:2} into several steps.
\vskip0.4cm\par\noindent
\underline{\it First step}
First we consider non-zero elements $s^\pm\in (\Sigma_+\otimes\Delta_\pm)\cap\mathbb S'$ as in 
\eqref{eq:first-step-elements} and use \eqref{eq:crucialevaluation}. In this case, the restriction of $\gamma^\varphi$ to $\mathbb C\be_0\wedge\Lambda^1\mathbb E\wedge\Lambda^3\mathbb F$ has kernel $\mathbb C\be_0\wedge\langle\be_1,\be_2\rangle\wedge\Lambda^3\mathbb F$, which is nothing but the fixed point set of $A:=\be_{34}+\be_{56}\in\mathfrak{so}(\mathbb E)$. Therefore \eqref{eq:crucialevaluation} vanishes if and only if
$A\big(\widetilde\kappa_{\mathbb E}(\sigma_i,\varsigma_j)\big)=0$ for all $i,j=1,2$, which in turn reads
\vskip0.2cm\par\noindent
$$
A(\sigma_i\wedge\varsigma_j)=0\quad\text{for all}\;i,j=1,2\;,
$$
\vskip0.2cm\par\noindent
since 
$\widetilde\kappa_{\mathbb E}:\Lambda^2\Sigma_+\to\mathbb E$ is an isomorphism. 

The kernel of $A$ acting on $\Lambda^2\Sigma_+$ is $\langle\upsilon_+\wedge\upsilon_-,\omega_+\wedge\omega_-\rangle$ thanks to Lemma \ref{lemma:eigenvectorsbasis} and we just showed that the {\it decomposable} $\sigma_i\wedge\varsigma_j$ is an element therein. Therefore $\sigma_i\wedge\varsigma_j$ vanishes or it is either proportional to $\upsilon_+\wedge\upsilon_-$ or proportional to $\omega_+\wedge\omega_-$, for all $i,j=1,2$. We thus reach the following dichotomy: we either have
\begin{equation}
\label{eq:firsttype}
\begin{aligned}
(\Sigma_+\otimes\Delta_\pm)\cap\mathbb S'&\subset \Upsilon\otimes\Delta_\pm\;,\\
(\Sigma_-\otimes\Delta_\pm)\cap\mathbb S'&\subset J\Upsilon\otimes\Delta_\pm\;,
\end{aligned}
\end{equation}
or
\begin{equation}
\label{eq:secondtype}
\begin{aligned}
(\Sigma_+\otimes\Delta_\pm)\cap\mathbb S'&\subset \Omega\otimes\Delta_\pm\;,\\
(\Sigma_-\otimes\Delta_\pm)\cap\mathbb S'&\subset J\Omega\otimes\Delta_\pm\;,
\end{aligned}
\end{equation}
where we also used that $\mathbb S'$ is invariant under conjugation $c$.
\vskip0.4cm\par\noindent
\underline{\it Second step}
We now consider two elements $s\in(\Sigma_+\otimes\Delta_\pm)\cap\mathbb S'$ and $t\in (\Sigma_-\otimes\Delta_\pm)\cap\mathbb S'$ and set to establish that $\gamma^\varphi(s,t)\in\langle\be_{12},\be_{34},\be_{56}\rangle\oplus\Lambda^2 \mathbb F$.
(We note that $\kappa(s,t)\in\mathbb C\be_0$ is not necessarily zero; however, this fact does not play any role  here.)

Thanks to Proposition \ref{prop:imgamma} and Proposition \ref{prop:imgammauseful}, it is sufficient to show that $\omega^{(2)}(s,t)$ lands in the kernel of the composition 
\begin{equation}
\label{eq:restrictionmap}
\pi\circ\gamma^\varphi|_{\Lambda^2\mathbb E}:\Lambda^2\mathbb E\longrightarrow\gamma^\varphi(\Lambda^2\mathbb E)\longrightarrow\frac{\gamma^\varphi(\Lambda^2\mathbb E)}{\gamma^\varphi(\Lambda^2\mathbb E)\cap\stab_{\mathfrak{so}(\mathbb E)}(\varphi)}
\end{equation}
of the restriction of the map $\gamma^\varphi$ to $\Lambda^2\mathbb E$ together with the natural projection from $\gamma^\varphi(\Lambda^2\mathbb E)$ to $\gamma^\varphi(\Lambda^2\mathbb E)\cap\stab_{\mathfrak{so}(\mathbb E)}(\varphi)=\langle\be_{12},\rho\be_{34}+\lambda\be_{56}\rangle$. The kernel of \eqref{eq:restrictionmap} consists of the kernel of $\gamma^\varphi|_{\Lambda^2\mathbb E}$ as detailed in (1) of Proposition \ref{prop:imgamma} together with $\rho\be_{34}+\lambda\be_{56}$, $\be_{12}$. All in all, we get $\langle\be_{35}, \be_{36},\be_{45},\be_{46},\be_{34},\be_{56},\be_{12}\rangle$, which is precisely the centralizer  of $\be_{12}$ in $\Lambda^2\mathbb E$. Since $s$ and $t$ are eigenvectors for $\be_{12}$ relative to opposite eigenvalues (thanks to \eqref{eq:firsttype}-\eqref{eq:secondtype} of the first step, Lemma \ref{lemma:eigenvectorsbasis}, and the fact that $I$ and $J$ anticommute), the claim of the second step is proved.
\vskip0.4cm\par\noindent
\underline{\it Third step}
If $\dim\big((\Sigma_i\otimes\Delta_j)\cap\mathbb S'\big)\leq 2$ for some $i,j\in\left\{+,-\right\}$, 
then $\dim\mathbb S'\leq 26$, arguing as in the beginning of the proof of Theorem \ref{mainthm:1}. In view of 
\eqref{eq:firsttype}-\eqref{eq:secondtype} as well, 
\vskip0.4cm\par\noindent
\begin{equation}
\label{eq:assumption-star}
\tag{$\star$}
\text{\it We assume from now on that}\;\; 3\leq\dim\big((\Sigma_i\otimes\Delta_j)\cap\mathbb S'\big)\leq 4\;\;\text{\it for all}\;\; 
i,j\in\left\{+,-\right\}
\end{equation}
\vskip0.5cm\par\noindent
\underline{\it Fourth step}
{\bf We first study the case \eqref{eq:firsttype} when $\rho=\lambda$} (so we assume w.l.o.g. $\rho=\lambda=1$). This separate analysis is due to the fact that
$\varphi$ acts trivially on $(\Upsilon\oplus J\Upsilon)\otimes\Delta$, according to Lemma \ref{lemma:eigenvectorsbasis}. Therefore, from \eqref{lem:reformulation} of Lemma \ref{lem:mitoccaaggiungere}, we see that
$
\eta(\gamma^\varphi(s,t)v,w)=
\tfrac12\langle s,\imath_v\imath_w\varphi\cdot t\rangle
$
for all $s,t$ in the spaces on the left of \eqref{eq:firsttype}. 

Thanks to the second step, the only possible 
non-zero contributions may occur when $s\in(\Sigma_+\otimes\Delta_\pm)\cap\mathbb S'$, $t\in (\Sigma_-\otimes\Delta_\pm)\cap\mathbb S'$, and $v\wedge w\in\left\{\be_{12},\be_{34},\be_{56}\right\}$:
\vskip0.2cm\par\noindent
\begin{enumerate}
\item If $v\wedge w=\be_{12}$, then $\imath_v\imath_w\varphi=-(\be_{34}+\be_{56})$ and $\imath_v\imath_w\varphi\cdot t=0$ again by Lemma \ref{lemma:eigenvectorsbasis};
\item If $v\wedge w=\be_{34}$, then $\imath_v\imath_w\varphi=-\be_{12}$ and 
\begin{equation}
\begin{aligned}
\label{eq:nonlosofran}
\eta(\gamma^\varphi(s,t)\be_3,\be_4)&=-\tfrac12\langle s,\be_{12}\cdot t\rangle=\tfrac12\langle s,I\cdot t\rangle\\
&=i\langle s,t\rangle\;,
\end{aligned}
\end{equation}
\item If $v\wedge w=\be_{56}$, then $\eta(\gamma^\varphi(s,t)\be_5,\be_6)=i\langle s,t\rangle$ exactly as above.
\end{enumerate}
\vskip0.2cm\par\noindent
In summary the only possible non-zero output of $\gamma^\varphi$ so far is given by the element $\be_{34}+\be_{56}$. Since
$3\leq\dim\big((\Sigma_i\otimes\Delta_j)\cap\mathbb S'\big)\leq 4$ for all $i,j\in\left\{+,-\right\}$, 
this can be achieved choosing
$s^\pm\in (\Sigma_+\otimes\Delta_\pm)\cap\mathbb S'$ and $t^\pm\in(\Sigma_-\otimes\Delta_\pm)\cap\mathbb S'$ so that
\vskip0.2cm\par\noindent
\begin{equation}
\begin{aligned}
\langle s^+,t^+\rangle&=-i\Longrightarrow\gamma^\varphi(s^+,t^+)=\be_{34}+\be_{56}\;,\\
\langle s^-,t^-\rangle&=-i\Longrightarrow\gamma^\varphi(s^-,t^-)=\be_{34}+\be_{56}\;.
\end{aligned}
\end{equation}
\vskip0.2cm\par\noindent
Moreover, since $\langle-,-\rangle=f\otimes f$ while the $\mathbb C\be_0$-component of the Dirac current is $f_{\vol_E}\otimes f_{\vol_F}$, we can check that $s^+\odot t^+ + s^-\odot t^-$ is an element of the (complexified) Dirac kernel $\mathfrak D\otimes_{\mathbb R}\mathbb C$.
We thus arrived at the following rather rewarding result.
\begin{Proposition}
\label{prop:e34+e56}
If $\mu=0$, \eqref{eq:firsttype} holds with $\rho=\lambda$, and \eqref{eq:assumption-star} is in force, then the stabilizer Lie algebra $\mathfrak h$ includes $\be_{34}+\be_{56}$. In particular, $\mathbb S'$ is compatible with the eigenspace decomposition of $\be_{34}+\be_{56}$ on $\mathbb S$. 
\end{Proposition}
\vskip0.4cm\par\noindent
\underline{\it Fifth step}
The eigenspaces of $\be_{34}+\be_{56}$ on $\mathbb S$ are the $16$-dimensional $\mathbb S_0:=(\Upsilon\oplus J\Upsilon)\otimes\Delta$ (w.r.t. eigenvalue $0$), and the $8$-dimensional $\mathbb S_{\pm 2i}:=\langle\omega_\pm,J\omega_\mp\rangle\otimes\Delta$
(w.r.t. eigenvalue $\pm 2I$). We note that $\mathbb S_{\pm 2i}=\mathbb S_{\pm 2i,+}\oplus \mathbb S_{\pm 2i,-}$w.r.t. to the decomposition $\Sigma=\Sigma_+\oplus\Sigma_-$, where
\vskip0.2cm\par\noindent
\begin{equation}
\label{eq:decompfinal}
\mathbb S_{\pm 2i,+}:=\mathbb C\omega_\pm\otimes\Delta\;,\qquad \mathbb S_{\pm 2i,-}:=\mathbb C J\omega_\mp\otimes\Delta\;.
\end{equation}
\vskip0.2cm\par\noindent
If we assume that $\dim\mathbb S'> 24$, then 
Proposition \ref{prop:e34+e56} and the fact that $\mathbb S'$ is invariant under the conjugation $c$ tell us that 
$\dim(\mathbb S'\cap \mathbb S_{\pm 2i} )\geq 5$. However both spaces in \eqref{eq:decompfinal} are $4$-dimensional, so $\dim(\mathbb S'\cap \mathbb S_{\pm 2i,+} )\geq 1$ and $\dim(\mathbb S'\cap \mathbb S_{\pm 2i,-} )\geq 1$. We may then consider $s=\omega_+\otimes\epsilon\in\mathbb S'$, for some $\epsilon=\epsilon^++\epsilon^-\in\Delta=\Delta_+\oplus\Delta_-$, with, say, $\epsilon^+\neq 0$. It is convenient to write $s=s^++s^-$, where $s^\pm=\omega_+\otimes\epsilon^\pm$.

Now take
$t\in (\Sigma_-\otimes\Delta_+)\cap\mathbb S'\subset J\Upsilon\otimes\Delta_+$ and note that
\vskip0.2cm\par\noindent
$$\gamma^\varphi(s^-,t)=\kappa(s^-,t)=\kappa(s^+,t)=0\;,$$
\vskip0.2cm\par\noindent
because the element $s^\pm\odot t$ is an eigenvector for the operator $\be_{34}+\be_{56}$ with eigenvalue $2I$, while $\gamma^\varphi(s^-,t)\in\mathbb C\be_0\wedge\mathbb F$, $\kappa(s^-,t)\in\mathbb F$, and $\kappa(s^+,t)\in\mathbb C\be_0$. Thus
$s\odot t$ is an element of the Dirac kernel and
\vskip0.2cm\par\noindent
\begin{equation*}
\begin{aligned}
\gamma^\varphi(s,t)&=\gamma^\varphi(s^+,t)\in\Lambda^2\mathbb E\oplus\Lambda^2\mathbb F\;,\\
\end{aligned}
\end{equation*}
\vskip0.2cm\par\noindent
with the component in $\Lambda^2\mathbb F$ absent, again by the eigenvalue structure of the operator $\be_{34}+\be_{56}$. 
Finally
\vskip0.2cm\par\noindent
\begin{enumerate}
	\item we may choose $t$ so that
$\omega^{(2)}_{\mathbb E}\otimes f$ evaluated on $s^+\odot t$ is non-zero (this is true because $\dim\big((\Sigma_-\otimes\Delta_+)\cap\mathbb S'\big)\geq 3$ and $\omega^{(2)}_{\mathbb E}$ does not vanish on the decomposable tensors),
	\item the element $s^+\odot t$ is an eigenvector for the operator $\be_{12}$ with the eigenvalue $-2I$, 
\end{enumerate}  
\vskip0.2cm\par\noindent
so $s^+\odot t$ does not belong to the kernel of the map \eqref{eq:restrictionmap}. Summing things up, we obtained an element $s\odot t$ inside the Dirac kernel with the property that $\gamma^\varphi(s,t)$ does {\it not} stabilize $\varphi$.
As already advertised, {\it this is a contradiction and it shows that $\dim\mathbb S'\leq 24$}.
\vskip0.4cm\par\noindent
\underline{\it Sixth step}
{\bf We study the case \eqref{eq:firsttype} when $\rho\neq\lambda$ and the case \eqref{eq:secondtype} simultaneously}. 
We note that Clifford multiplication by $\varphi$ is injective on   
$(\Sigma_i\otimes\Delta_j)\cap\mathbb S'$ for all $i,j\in\left\{+,-\right\}$.
For concreteness of exposition, we will now treat the case \eqref{eq:firsttype} when $\rho\neq\lambda$, but all of our arguments work equally well in the case \eqref{eq:secondtype}.

We consider two elements $s\in(\Sigma_+\otimes\Delta_\pm)\cap\mathbb S'$ and $t\in (\Sigma_-\otimes\Delta_\pm)\cap\mathbb S'$ as in the second step and focus on the component in $\Lambda^2\mathbb F$ of $\gamma^\varphi(s,t)\in\langle\be_{12},\be_{34},\be_{56}\rangle\oplus\Lambda^2\mathbb F$. Now $\mathfrak{so}(\mathbb F)\cong\Lambda^2 \mathbb F$ and the decomposition $\Lambda^2 \mathbb F=\Lambda^2 _+ \mathbb F\oplus\Lambda^2_-\mathbb F$ into self-dual and anti-self-dual forms corresponds to the obvious splitting of $\mathfrak{so}(\mathbb F)$ into two ideals. By applying \eqref{lem:reformulation} of Lemma \ref{lem:mitoccaaggiungere} with $v,w\in\mathbb F$, we see
$\eta(\gamma^\varphi(s,t)v,w)=-\tfrac{1}{12}\langle\varphi\cdot s,v\wedge w \cdot t\rangle
-\tfrac{1}{12}\langle\varphi\cdot t,v\wedge w \cdot s\rangle
=-\tfrac{1}{6}\langle\varphi\cdot s,v\wedge w \cdot t\rangle
$
and so
\vskip0.2cm\par\noindent
\begin{equation}
\label{eq:Lambda^2E}
\gamma^\varphi(s,t)=-\tfrac{1}{6}\big(f_{\varphi}\otimes\omega^{(2)}_{\mathbb F}\big)(s,t)\;,
\end{equation} 
\vskip0.2cm\par\noindent
where $f_\varphi$ is the composition of $f$ with $\varphi$ in the second entry, with a little abuse of notation.
Now the map $\omega^{(2)}_{\mathbb F}:\odot^2\Delta_\pm\to\Lambda^2_\pm\mathbb F$ is an isomorphism, whereas $f_\varphi$ is a pairing of $\Upsilon$ with $J\Upsilon$ (as advertised above, here we are using in a crucial way that $\rho\neq \lambda$).
\vskip0.2cm\par\noindent

The goal of this step is to prove the following result, which can be regarded as the substitute of 
the former Proposition \ref{prop:e34+e56}.
\begin{Proposition}
\label{prop:8}
If $\mu=0$, \eqref{eq:firsttype} holds with $\rho\neq\lambda$, and \eqref{eq:assumption-star} is in force, then the stabilizer Lie algebra $\mathfrak h$ includes the whole $\mathfrak{so}(F)$.
\end{Proposition}
\begin{proof}
By Corollary \ref{cor:hermitianforms}, we may fix a basis $\left\{\epsilon_1^+,\epsilon_2^+\right\}$ of $\Delta_{+}$ with the property that $$
\epsilon_2^+=J\epsilon_1^+\qquad\text{and}\qquad f(\epsilon_1^+,\epsilon_2^+)=+1\;.$$ 
In particular the isomorphism $\odot^2\Delta_+\cong\Lambda^2_+\mathbb F\cong\mathfrak{sl}_2(\mathbb C)$ reads as
\begin{equation*}
\begin{aligned}
\tfrac12\epsilon_1^+\epsilon_1^+&\cong E:=\begin{pmatrix}0 & 1 \\ 0 & 0\end{pmatrix}\;.\\
-\tfrac12\epsilon_2^+\epsilon_2^+&\cong F:=\begin{pmatrix}0 & 0 \\ 1 & 0\end{pmatrix}\;.\\
-\epsilon_1^+\epsilon_2^+&\cong H:=\begin{pmatrix}1 & 0 \\ 0 & -1\end{pmatrix}\;.
\end{aligned}
\end{equation*}
Since $3\leq\dim\big((\Sigma_+\otimes\Delta_+)\cap\mathbb S'\big)\leq 4$, we may assume, up to a reordering of the basis elements of $\Sigma_+$ and $\Delta_+$, that the space $(\Sigma_+\otimes\Delta_+)\cap\mathbb S'$ includes at least three vectors of the form
\begin{equation}
\label{eq:++elements}
\begin{aligned}
w_1&=\upsilon_+\otimes \epsilon_1^+ +a \upsilon_-\otimes\epsilon_2^+\;,\\
w_2&=\upsilon_+\otimes \epsilon_2^+ + b \upsilon_-\otimes \epsilon_2^+\;,\\
w_3&=\upsilon_-\otimes \epsilon_1^++c \upsilon_-\otimes\epsilon_2^+\;,
\end{aligned}
\end{equation}
for some $a,b,c\in\mathbb C$. Since $\mathbb S'$ is invariant under conjugation, the space $(\Sigma_-\otimes\Delta_+)\cap\mathbb S'$ includes
\begin{equation}
\label{eq:-+elements}
\begin{aligned}
c(w_1)&=J\upsilon_+\otimes \epsilon_2^+ -\overline{a} J\upsilon_-\otimes\epsilon_1^+\;,\\
c(w_2)&=-J\upsilon_+\otimes \epsilon_1^+ - \overline{b} J\upsilon_-\otimes \epsilon_1^+\;,\\
c(w_3)&=J\upsilon_-\otimes \epsilon_2^+ -\overline{c} J\upsilon_-\otimes\epsilon_1^+\;.
\end{aligned}
\end{equation}
Using \eqref{eq:Lambda^2E} and the action of $\varphi$ detailed in Lemma \ref{lemma:eigenvectorsbasis}, we may tabulate the contribution in $\mathfrak{so}(\mathbb F)\cong\Lambda^2\mathbb F$ arising from pairing elements \eqref{eq:++elements} with \eqref{eq:-+elements}.  For convenience of exposition, we omit the non-zero multiplicative factor $\tfrac{i}{6}(\rho-\lambda)$ that appears overall.
\vskip0.2cm\par\noindent
{\footnotesize
\begin{align*}
 \begin{array}{|c|c|c|c|} \hline
\Lambda^2_+\mathbb F  &c(w_1) & c(w_2)  & c(w_3)\\ \hline\hline
 w_1& \begin{gathered} \\ -(1+|a|^2) \epsilon_1^+\epsilon_2^+ \\ \\ \end{gathered}& 
\begin{gathered} 
\\
\epsilon_1^+\epsilon_1^+          
-a\overline{b} \epsilon_1^+\epsilon_2^+
\\
\\
\end{gathered}  &
\begin{gathered}
\\
a \epsilon_2^+\epsilon_2^+              
 -a\overline{c}   \epsilon_1^+  \epsilon_2^+ 
\\
\\
\end{gathered}
\\ \hline
 w_2& \begin{gathered} 
\\
-\epsilon_2^+\epsilon_2^+          
-\overline{a}b \epsilon_1^+\epsilon_2^+
\\
\\
\end{gathered} & \begin{gathered}
\\
(1- |b|^2)  \epsilon_1^+\epsilon_2^+
\\
\\
\end{gathered}
 &
\begin{gathered}\\
 b  \epsilon_2^+ \epsilon_2^+
-b\overline{c}    \epsilon_1^+\epsilon_2^+
\\
\\
\end{gathered}
\\ \hline
 w_3& \begin{gathered}
\\
-\overline{a} \epsilon_1^+\epsilon_1^+              
 -\overline{a}c   \epsilon_1^+  \epsilon_2^+ 
\\
\\
\end{gathered}	 & \begin{gathered}\\
 -\overline{b}  \epsilon_1^+ \epsilon_1^+
-\overline{b}c    \epsilon_1^+\epsilon_2^+
\\
\\
\end{gathered} &
\begin{gathered}
\\
(1-|c|^2)\epsilon_1^+\epsilon_2^+ 
-\overline{c} \epsilon_1^+\epsilon_1^+
+c\epsilon_2^+\epsilon_2^+ 
\\
\\
\end{gathered}\\ \hline
 \end{array}
 \end{align*}
}
\vskip0.2cm\par\noindent\par\noindent
Similarly, we provide the contribution coming from the Dirac current, which is proportional to $\be_0$. We write the various components along $\be_0$ up to an overall sign factor, for simplicity.

{\footnotesize
\begin{align*}
 \begin{array}{|c|c|c|c|} \hline
\kappa  &c(w_1) & c(w_2)  & c(w_3)\\ \hline\hline
 w_1& \begin{gathered} \\ 1+|a|^2 \\ \\ \end{gathered}& 
\begin{gathered} 
\\
a\overline{b} 
\\
\\
\end{gathered}  &
\begin{gathered}
\\
a\overline{c}   
\\
\\
\end{gathered}
\\ \hline
 w_2& \begin{gathered} 
\\
\overline{a}b 
\\
\\
\end{gathered} & \begin{gathered}
\\
1+ |b|^2  
\\
\\
\end{gathered}
 &
\begin{gathered}\\
 b\overline{c}    
\\
\\
\end{gathered}
\\ \hline
 w_3& \begin{gathered}
\\           
 \overline{a}c  
\\
\\
\end{gathered}	 & \begin{gathered}\\
\overline{b}c   
\\
\\
\end{gathered} &
\begin{gathered}
\\
1+|c|^2
\\
\\
\end{gathered}\\ \hline
 \end{array}
 \end{align*}
}
\par\noindent
The complex line generated by the bispinor $$\frac{w_1\otimes c(w_1)}{1+|a|^2}-\frac{w_2\otimes c(w_2)}{1+|b|^2}$$
is in the kernel of the Dirac current operator and its image in $\Lambda^2_+\mathbb F$ via the map \eqref{eq:Lambda^2E} is $\mathbb C\epsilon_1^+\epsilon_2^+$. If $ab=0$, then we may also consider the two bispinors $w_1\otimes c(w_2)$ and $w_2\otimes c(w_1)$, thus we get all of $\Lambda^2_+\mathbb F$. The same holds if $ab\neq 0$ and $c=0$, by considering the two bispinors $w_3\otimes c(w_1)$ and $w_1\otimes c(w_3)$. Finally, if $abc\neq 0$, the complex line generated by
$$\frac{w_3\otimes c(w_3)}{1+|c|^2}-\frac{w_1\otimes c(w_1)}{1+|a|^2}\;,$$
is in the kernel of the Dirac current and its image in $\Lambda^2_+\mathbb F$ is $\mathbb C(\overline{c}\epsilon_1^+\epsilon_1^+ -c\epsilon_2^+\epsilon_2^+)$.

According to the second step, we have therefore established that the stabilizer $\mathfrak h$ includes at least
$H,E+e^{i\theta}F$ for some $\theta\in\mathbb R$,
modulo $\langle\be_{12},\be_{34},\be_{56}\rangle\subset\Lambda^2\mathbb  E$.
Since $\mathfrak h$ is closed under Lie brackets, we see that $\mathfrak h$ includes exactly the whole $\Lambda^2_+\mathbb F$.

We may repeat the same argument using the spaces $(\Sigma_+\otimes\Delta_-)\cap\mathbb S'$ and $(\Sigma_-\otimes\Delta_-)\cap\mathbb S'$ to get $\mathfrak h\supset \Lambda^2_-\mathbb F$. This proves the proposition. 
\end{proof}
Since the proof of the following result is as for Proposition \ref{prop:8}, we omit it.
\begin{Proposition}
\label{prop:9}
If $\mu=0$, \eqref{eq:secondtype} holds, and \eqref{eq:assumption-star} is in force, then the stabilizer Lie algebra $\mathfrak h$ includes the whole $\mathfrak{so}(F)$.
\end{Proposition}
 \begin{Corollary}
\label{cor:5}
If $\mu=0$, \eqref{eq:firsttype} holds with $\rho\neq\lambda$ or \eqref{eq:secondtype} holds, and \eqref{eq:assumption-star} is in force, then the subspace $\mathbb S'$ of $\mathbb S$ is compatible with the decomposition $\mathbb S=\big(\Sigma\otimes\Delta_+\big)\oplus\big(\Sigma\otimes\Delta_-\big)$ of $\mathbb S$.
\end{Corollary}
\begin{proof}
Immediate from the fact that $\mathbb S'$ is stable under $\mathfrak{so}(F)$ by Propositions \ref{prop:8} and \ref{prop:9}, and the fact that $\Delta_+$ and $\Delta_-$ are irreducible and inequivalent under the action of  $\mathfrak{so}(F)$.
\end{proof} 
\vskip0.1cm\par\noindent
\underline{\it Seventh and last step}
 We are now ready to complete the proof of the main Theorem \ref{mainthm:2}.
\vskip0.3cm\par

Case \eqref{eq:firsttype} with $\rho=\lambda$ has already been settled in the fifth step. Thus we may assume that the assumptions of Corollary \ref{cor:5} hold and $\mathfrak h\supset \mathfrak{so}(F)$. 
In particular $\dim\big((\Sigma_i\otimes\Delta_j)\cap\mathbb S'\big)=4$ for all $i,j\in\left\{+,-\right\}$ and
either
\begin{equation}
\label{eq:firsttypefinal}
\begin{aligned}
(\Sigma_+\otimes\Delta_\pm)\cap\mathbb S'&=\Upsilon\otimes\Delta_\pm\;,\\
(\Sigma_-\otimes\Delta_\pm)\cap\mathbb S'&=J\Upsilon\otimes\Delta_\pm\;,
\end{aligned}
\end{equation}
or
\begin{equation}
\label{eq:secondtypefinal}
\begin{aligned}
(\Sigma_+\otimes\Delta_\pm)\cap\mathbb S'&=\Omega\otimes\Delta_\pm\;,\\
(\Sigma_-\otimes\Delta_\pm)\cap\mathbb S'&=J\Omega\otimes\Delta_\pm\;,
\end{aligned}
\end{equation}
thanks to the second step. Let us deal with \eqref{eq:firsttypefinal} just for concreteness.

Since $(M,g,F)$ is highly supersymmetric and $\mathbb S'\supset  
(\Upsilon\oplus J\Upsilon)\otimes\Delta$, there exists a non-zero element
$s\in\mathbb S'$ that lies into $(\Omega\oplus J\Omega)\otimes\Delta$. By Corollary \ref{cor:5}, we may assume that it lies into $(\Omega\oplus J\Omega)\otimes\Delta_+$ (the alternative case $(\Omega\oplus J\Omega)\otimes\Delta_-$ is analogous) and we write $s=s_++s_-$ according to the decomposition $\Sigma=\Sigma_+\oplus\Sigma_-$. 
Then
\begin{equation*}
\begin{aligned}
\kappa(s,\Upsilon\otimes\Delta_-)=\kappa(s,J\Upsilon\otimes\Delta_-)\in\mathbb F\;,
\end{aligned}
\end{equation*}
thanks to \eqref{eq:Diraccurrentbroken} and the fact that $f_{\vol_F}$ has the invariant $\imath=+1$. However, by the eigenvalue structure of Lemma \ref{lemma:eigenvectorsbasis} and the fact that $\be_{12},\be_{34},\be_{56}$ act trivially on $\mathbb F$ we finally see that both $\kappa(s,\Upsilon\otimes\Delta_-)$ and $\kappa(s,J\Upsilon\otimes\Delta_-)$ do vanish. In summary the spaces
$$
\mathbb Cs\otimes \big(\Upsilon\otimes\Delta_-)\;,\qquad \mathbb Cs\otimes \big(J\Upsilon\otimes\Delta_-\big)\;,
$$
are included in the Dirac kernel. If $s_+\neq 0$, then 
\vskip0.2cm\par\noindent
\begin{equation}
\label{eq:celafacciamo}
\widetilde\kappa_{\mathbb E}\otimes\omega^{(3)}_{\mathbb F}(s,\Upsilon\otimes\Delta_-)=\widetilde\kappa_{\mathbb E}\otimes\omega^{(3)}_{\mathbb F}(s_+,\Upsilon\otimes\Delta_-)
\end{equation}
\vskip0.2cm\par\noindent
and therefore the components in $\mathbb E\wedge \mathbb F$ of $\gamma^\varphi(s,\Upsilon\otimes\Delta_-)$ and $\gamma^\varphi(s_+,\Upsilon\otimes\Delta_-)$ coincide. Since $s_+\in\Omega\otimes\Delta_+$, the input on the r.h.s. of  
\eqref{eq:celafacciamo} is not annihilated by $\be_{34}+\be_{56}$ and we then get a non-trivial contribution in $\mathbb E\wedge \mathbb F$, according to the initial lines of the first step. {\it This is a contradiction and it shows that this case cannot happen in the highly-supersymmetric regime}.

If $s_+=0$ then $s_-\neq 0$, and the same conclusion is obtained using the space $J\Upsilon\otimes\Delta_-$ instead of the space $\Upsilon\otimes\Delta_-$.

\section*{Acknowledgments}
The third author 
acknowledges the MIUR Excellence Department Project MatMod@TOV, which has been
awarded to the Department of Mathematics, University of Rome Tor Vergata, CUP E83C23000330006.  This article/publication was also supported by the ``National Group for
Algebraic and Geometric Structures, and their Applications'' GNSAGA-INdAM (Italy) and it is
based upon work from COST Action CaLISTA CA21109 supported by COST
(European Cooperation in Science and Technology), 
{\scriptsize\url{https://www.cost.eu}}.
%

\bibliographystyle{abbrv}

\end{document}